\documentclass[11pt]{article}

\usepackage{extras-alg}
\usepackage{multirow}
\usepackage{appendix}
\usepackage{graphicx}
\usepackage{epstopdf}
\usepackage{amssymb,amsmath}

\newcommand{\e}{\mbox{\rm evenlevel}}
\renewcommand{\o}{\mbox{\rm oddlevel}}
\newcommand{\mn}{\mbox{\rm minlevel}}
\newcommand{\mx}{\mbox{\rm maxlevel}}
\renewcommand{\t}{\mbox{\rm tenacity}}
\newcommand{\bt}{\mbox{$\CB_{b,t}$}}
\newcommand{\base}{\mbox{\rm base}}
\newcommand{\CB}{\mbox{${\mathcal B}$}}
\newcommand{\lv}{\mbox{\rm level}}
\newcommand{\bd}{\mbox{\rm bud}}
\newcommand{\bds}{\mbox{${\rm bud^*}$}}

\newcommand{\suppress}[1]{}

\newcommand{\la}{\leftarrow}

\newcommand{\Z}{\mbox{\rm\bf Z}}
\newcommand{\Zplus}{\Z^+}

\makeatletter
\def\fnum@figure{{\bf Figure \thefigure}}
\def\fnum@table{{\bf Table \thetable}}

\long\def\@mycaption#1[#2]#3{\addcontentsline{\csname
 ext@#1\endcsname}{#1}{\protect\numberline{\csname
  the#1\endcsname}{\ignorespaces #2}}\par
     \begingroup
       \@parboxrestore
          \small
       \@makecaption{\csname fnum@#1\endcsname}{\ignorespaces
#3\endgroup}
      }

\def\Remark{\par\noindent{\bf Remark}:\enspace}

\newcommand{\PP}{\mbox{\rm \bf P}}
\newcommand{\SP}{\mbox{\rm \bf \# P}}

\begin{document}

\title{A Simplification of the MV Matching Algorithm and its Proof}

\author{
{\Large\em Vijay V. Vazirani}\thanks{Supported by NSF Grants CCF-0914732 and CCF-1216019, and a Guggenheim Fellowship.
College of Computing, Georgia Institute of Technology, Atlanta, GA 30332--0280,
E-mail: {\sf vazirani@cc.gatech.edu}. 
}
}

\maketitle

\begin{abstract}
For all practical purposes, the Micali-Vazirani \cite{MV} general graph maximum matching algorithm is still the most efficient known algorithm
for the problem. The purpose of this paper is to provide a complete proof of correctness of the algorithm
in the simplest possible terms; graph-theoretic machinery developed for this purpose also helps simplify the algorithm.
\end{abstract}

\section{Introduction}
\label{sec.intro}

For all practical purposes, the Micali-Vazirani \cite{MV} general graph maximum matching algorithm is still the most efficient known algorithm
for the problem; see Section \ref{sec.running} for precise details. The process of proving correctness of this algorithm was started in \cite{va.matching};
however, as detailed in Section \ref{sec.overview}, the paper had deficiencies. The purpose of the current paper is to provide a complete proof 
of this algorithm in the simplest possible terms; graph-theoretic machinery developed for this purpose also helps simplify the algorithm.

Matching occupies a central place in the theory of algorithms as detailed below. So it is quite unfortunate that non-bipartite matching is
viewed today as being ``complicated'' when in reality is it elegant, and of course, immensely useful in numerous situations.
The fact that many of the algorithmic papers dealing with this topic leave a lot to be
desired in terms of correctness, completeness and clarity, is no doubt a major reason for this perception -- on the flip side, it is important to
mention that non-bipartite matching is a difficult topic and getting it right first time is not easy. In this paper, we have
pointed out these shortcomings in a dispassionate manner in order to lead to clarity and remove misconceptions, see Section
\ref{sec.running} and \ref{sec.overview}. This paper is fully self-contained so as to be suitable for pedagogical and archival purposes.

From the viewpoint of efficient algorithms, bipartite and non-bipartite matching problems are qualitatively different. First consider the process of
finding a maximum matching by repeatedly finding augmenting paths. Whereas the in former case, all alternating paths from an unmatched vertex 
$f$ to a matched vertex $v$ must have the same parity, even or odd, in the latter they can be of both parities. Edmonds defined the key 
notion of {\em blossoms} and finessed this difficulty in non-bipartite graphs by ``shrinking'' blossoms. 

The most efficient known maximum matching algorithms in both bipartite and non-bipartite graphs resort to finding minimum length augmenting 
paths  w.r.t. the current matching. However, from this perspective, the difference between the two classes of graphs becomes even more pronounced.
Unlike the bipartite case, in non-bipartite graphs minimum length alternating paths do not possess an elementary property, called 
breadth first search honesty\footnote{Intuitively, it states that in order to find shortest alternating paths from an unmatched vertex $f$ to all other
vertices, there is never a need to find a longer path to any vertex $v$.} in Section \ref{sec.tenacity}.
Indeed, in the face of this debilitating shortcoming, the problem of finding minimum length alternating paths appears to be intractable. It is 
a testament to the remarkable structural properties of matching that despite this, a near-linear time algorithm is possible.

Edmonds' blossoms are not adequate for the task of  finding minimum length augmenting paths since by ``shrinking'' them, length information 
is completely lost. What is needed is a definition of blossoms from the perspective of minimum length alternating paths, as given in \cite{va.matching}
and simplified in the current paper. This requires a substantial graph-theoretic development.  Indeed, once the reader is sufficiently comfortable 
with this structure, they will be able to appreciate how the algorithm harmoniously blends into it, thereby leading to a high level, conceptual picture. 
It is precisely this vantage point that enabled us to simplify the algorithm.

Matching has had a long and distinguished history spanning more than a century. The following quote from Lovasz and Plummer's
classic book \cite{LP.book} is most revealing:
\begin{quote}
Matching theory serves as an archetypal example of how a ``well-solvable'' problem can be studied. ... [It] is a central part of 
graph theory, not only because of its applications, but also because it is the source of important ideas developed during the rapid
growth of combinatorics during the last several decades. 
\end{quote}
Interestingly enough, matching has played an equally central role in the development of the theory of algorithms -- time and again, its
study has not only yielded powerful tools that have benefited other problems but also quintessential paradigms for the entire field. Examples
of the latter include the primal-dual paradigm \cite{Kuh55}, the definitions of the classes $\PP$ \cite{Edm.paths} and $\SP$ \cite{Val.permanent}, 
and the equivalence of random generation and approximate counting for self-reducible problems \cite{count.JVV}. Examples of the former include 
the notion of an augmenting path \cite{Konig,Egervary}, a method for determining the defining inequalities of the convex hull of
solutions to a combinatorial problem \cite{Ed.poly}, the canonical paths argument for showing expansion of the underlying graph of 
a Markov chain \cite{JS89}, and the Isolating Lemma \cite{MVV}. And at the interface of algorithms and game theory lies another highly 
influential matching algorithm: the Nobel Prize winning stable matching algorithm of Gale and Shapley \cite{GaleShapley}.

\subsection{Running time and related papers}
\label{sec.running}

The paper \cite{MV} had claimed a running time of $O(m \sqrt{n})$, on the pointer model, for finding a maximum matching 
in general graphs. However, this was based on an unproven claim that a certain datastructure task could be accomplished in linear time; see
Section \ref{sec.proof} for details. A very recent result \cite{PV} shows that the avenue suggested in \cite{MV} for proving this claim 
will not work. 

The current status is that the MV algorithm achieves a running time of $O(m \sqrt{n} \cdot \alpha(m, n))$ on the pointer model (using Tarjan's set 
union algorithm \cite{Tarjan}), where $\alpha$ is the inverse Ackerman function, and $O(m \sqrt{n})$ on the RAM model (using Gabow and 
Tarjan's linear time algorithm for a special case of set union \cite{GTarjan}); observe that \cite{GTarjan} appeared after \cite{MV}.
Since \cite{GTarjan} does not give a detailed explanation of why their idea is applicable to the MV algorithm, a new paper clarifying this
has been recently written by Gabow \cite{Gabow}.

We note that small theoretical improvements to the running time, for the case of
very dense graphs, have been given in recent years: $O(m \sqrt{n} {{\log (n^2/m)} / {\log n}})$ \cite{GKarzanov}  
and $O(n^w)$ \cite{Mucha}, where $w$ is the best exponent of $n$ for multiplication of
two $n \times n$ matrices. The former improves on MV for $m = n^{2 -o(1)}$ and the latter for $m = \omega(n^{1.85})$; additionally,
the latter algorithm involves a large multiplicative constant in its running time due to the use of fast matrix multiplication.

Over the years,
\cite{Blum} and \cite{GTarjan2} have claimed algorithms having the same running time as MV, and we need to clairfy the status of these
works. The first paper is simply wrong. It claims, without giving any details, that a certain procudure, called MBFS, runs in 
linear time; however, the requirements on this procedure are such that even a polynomial time implementation is not clear.
The second paper gives an efficient scaling algorithm for finding a minimum weight matching in a general graph with integral edge
weights. It claims that the unit weight version of their algorithm achieves the same running time as MV, but no details are provided
and no proof is offered. The question that arises in the reader's mind is why does it make sense to solve cardinality matching by reducing it to 
the harder problem of weighted matching -- no high level reason is given for this either.
The rest of the history of matching algorithms is very well documented and will not be repeated here, e.g., see \cite{LP.book,va.matching}.

\subsection{Overview and contributions of this paper}
\label{sec.overview}

To point out the contributions of this paper, we compare it to \cite{MV} and \cite{va.matching}.
\cite{MV} stated the matching algorithm in pseudocode -- this description is complete and correct. 
However, the paper did not provide a proof of correctness and running time.
The MV algorithm is quite elaborate and its pseudocode is extensive. For this reason, this was not an adequate way of expounding the algorithm.

The key to a conceptual description, as well as a proof of correctness, lay in formalizing the graph-theoretic structure underlying the
algorithm. This process was started in \cite{va.matching}; however, even this attempt turned out to be inadequate.
In retrospect, the proof and description of algorithm given in that paper had several shortcomings.
The central definition, that of blossoms, from the perspective of minimum length alternating paths, was too cumbersome.
As a consequence, some key facts, such as Theorems 5 and 6 in \cite{va.matching} were proven using complicated case analyses,
which don't even seem correct any more. A recent caerful reading revealed that several of the other proofs also have errors, even though the 
high level facts stated in the paper are by and large correct. The description of the algorithm, using these structural facts,
was also too cumbersome.

The process that was started in \cite{va.matching} has been brought to its logical conclusion in the current paper. A much simpler, recursive
definition of blossoms is presented in Section \ref{sec.blossom}. Blossoms, and their nesting structure, impose a strict regimen on how minimum 
length alternating paths traverse through them. These facts, an outcome of the new definition of blossoms, are established via conceptual proofs 
in Section \ref{sec.blossom}. These facts go to proving Theorem \ref{thm.bridge}, which is crucial for proving correctness of the algorithm; 
its full significance is discussed in Section \ref{sec.discussion}.

The MV algorithm involves two main ideas: a new search procedure called double depth first search (DDFS) and the precise synchronization of 
events. The former is described in Section \ref{sec.DDFS} and the latter in Section \ref{sec.alg}. The potential of finding other applications for
DDFS, as well as exploring variants and generalizations, remains inexplored so far. To facilitate wide dissemination, we 
have made Section \ref{sec.DDFS} fully self-contained. In this section, DDFS has been described in the simplified setting of a layered, 
directed graph and, unlike \cite{va.matching}, without resorting to any pseudocode. 
The description of the main algorithm, in Section \ref{sec.alg}, is given at a much higher level, e.g., without using low level notions 
such as ``anomaly edges'', as was done in \cite{va.matching}. 

Sections \ref{sec.tenacity} and  \ref{sec.base} give some basic structural properties of minimum length alternating paths which are needed for defining blossoms. 
Especially worth mentioning is Theorem \ref{thm.base}, which leads to the central notion of base of a vertex. This notion captures, in a simple manner, the essential 
aspect of what would have otherwise been an enormously complicated picture;
see the diverse-looking examples given in Theorem \ref{thm.base} to illustrate the various cases that can arise. 

Section \ref{sec.proof} ties up all the facts and completes the proof. Finally, equivalence of the two definitions of blossoms is established 
in Section \ref{sec.equivalence}. 
It is difficult to overemphasize the importance of well-chosen examples for understanding this result; indeed, 
most of the intuition lies in them and we have included several. Furthermore, they have been drawn in such a way that they easily reveal 
their structural properties (this involves drawing vertices in layers, according to their minlevel).

To the readers who are interested in quickly understanding the algorithm, we recommend that they dovetail appropriately between Sections \ref{sec.DDFS} 
and \ref{sec.alg}, and the rest of the paper; in particular, they will need definitions stated in Sections \ref{sec.tenacity}, \ref{sec.base}, \ref{sec.blossom} 
and \ref{sec.bridge}. Knowing the statements of theorems will also help; however, their proofs are not essential for this purpose.

\section{The tenacity of vertices and edges}
\label{sec.tenacity}

The MV algorithm finds augmenting paths in phases; in each {\em phase}, it finds a maximal set of disjoint minimum length
augmenting paths w.r.t. the current matching and it augments along all paths. It is easy to show that only $O(\sqrt{n})$ such phases
suffice for finding a maximum matching \cite{Karp,Karzanov}. The remaining task is designing an efficient algorithm for a phase. 
We embark on this task below.

All definitions are with respect to a matching $M$ in graph $G = (V, E)$. We will assume that $G$ has at least one unmatched vertex.
Throughout, $l_m$ will denote the length of a minimum length augmenting path in $G$; if $G$ has no augmenting paths, we will assume
that $l_m = \infty$.

\definition{(Evenlevel and oddlevel of vertices)}
The evenlevel (oddlevel) of a vertex $v$, denoted $\e(v)$ ($\o(v)$), is defined to be the length of a
minimum even (odd) length alternating path from an unmatched vertex to $v$; moreover, each such path will be called
an $\e(v)$ ($\o(v)$) path. If there is no such path, $\e(v)$ ($\o(v)$) is defined to be $\infty$. 

In all the figures, matched edges are drawn dotted, unmatched edges solid, and unmatched vertices are drawn with a small circle.
In Figure \ref{fig.vten}, evenlevels and oddlevels of vertices are indicated; missing levels are $\infty$.

\definition{(Maxlevel and minlevel of vertices)}
For a vertex $v$ such that at least one of $\e(v)$ and $\o(v)$ is finite, $\mx(v)$ ($\mn(v)$) is defined to be the bigger (smaller) of the two.

\begin{figure}[ht]
\begin{minipage}[b]{0.5\linewidth}
\centering
\includegraphics[width=\textwidth]{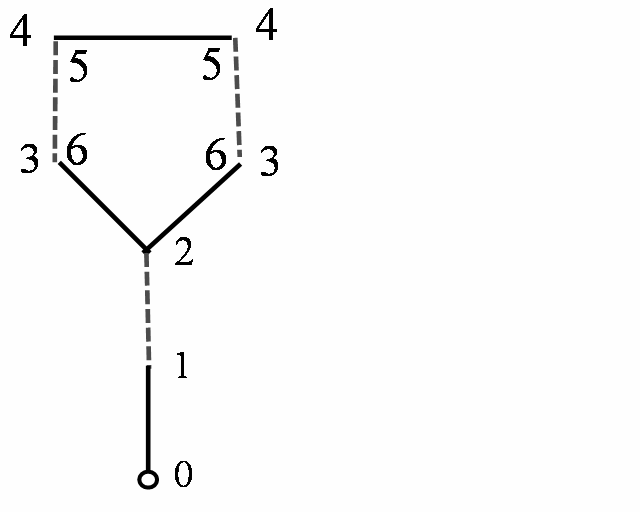}
\caption{
Evenlevels and oddlevels of vertices are indicated; missing levels are $\infty$.}
\label{fig.vten}
\end{minipage}
\hspace{0.5cm}
\begin{minipage}[b]{0.5\linewidth}
\centering
\includegraphics[width=\textwidth]{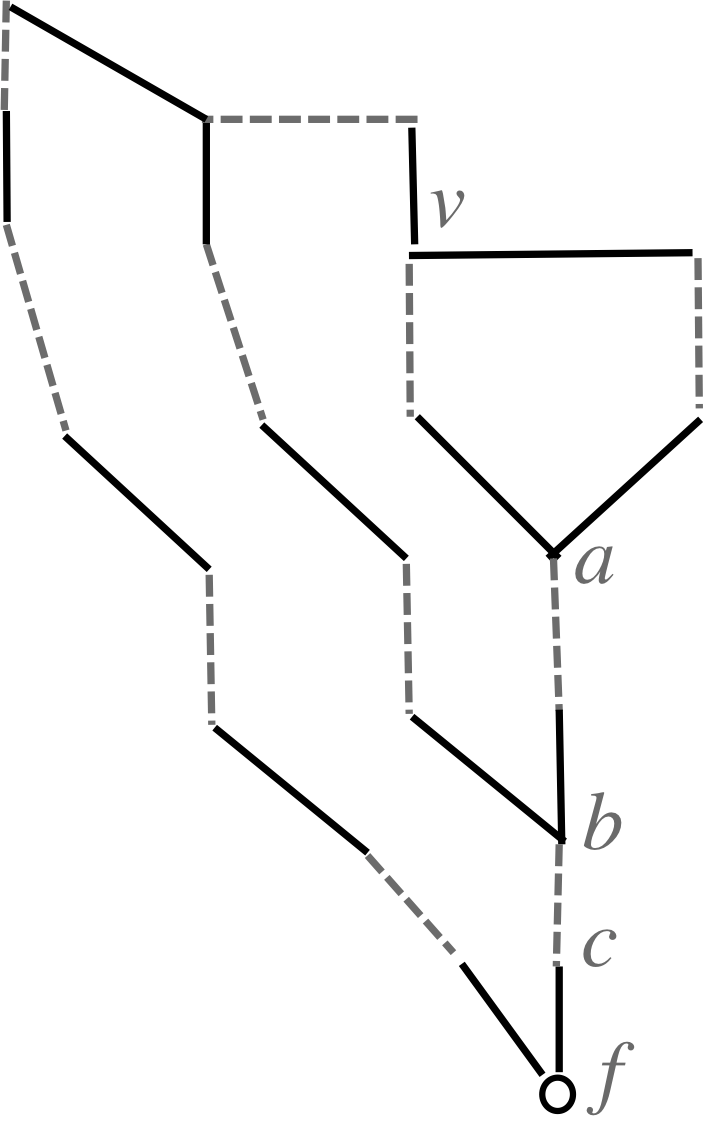}
\caption{Vertex $v$ is not BFS-honest on $\o(a)$ and $\e(c)$ paths.}
\label{fig.BFSH}
\end{minipage}
\end{figure}

In bipartite graphs, a vertex $v$ can have either an even or an odd length alternating path from an unmatched vertex $f$,
not both. Furthermore, minimum length alternating paths are {\em breadth first search honest} in the sense that
if $p$ is a minimum alternating path from $f$ to $v$ and $u$ lies on this path then 
$p[f \ \mbox{to} \ u]$\footnote{This denotes the part of $p$ from $f$ to $u$. Similarly 
$p[f \ \mbox{to} \ u)$ denotes the part of $p$ from $f$ to the vertex just before $u$, etc.} is a minimum alternating path from $f$ to 
$u$\footnote{Consequently, an alternating BFS suffices for executing a phase in bipartite graphs, e.g., see Section 2.1 in \cite{va.matching}.}.
This elementary property does not hold in non-bipartite graphs, e.g., if $v$ lies on the minimum alternating path from $f$ to $u$,
but with the opposite parity. This is precisely the reason that the task of finding minimum length augmenting paths is
considerably more difficult in non-bipartite graphs than in bipartite graphs.

In Figure \ref{fig.BFSH}, $v$ is not BFS-honest on $\o(a)$, $\o(b)$ and $\e(c)$ paths; it occurs at length 9 on the first and at
length 11 on the other two, even though $\o(v) = 7$. Thus shortest paths to $a$, $b$ and $c$, of appropriate parity, involve longer and 
longer odd length paths to $v$. Furthermore, this is not just an academic exercise: Suppose this graph had another edge $(c, d)$, where
$d$ is a new unmatched vertex. Then the only augmenting path uses the $\e(c)$ path. Whereas short paths are easy to find in a graph, finding
long paths is intractable, e.g. Hamiltonian path. Hence, at first sight, the problem of finding minimum length augmenting paths, which may
involve paths to certain vertices that are much longer than the shortest path, appears to be intractable in general graphs. As stated
in the Introduction, despite this, the remarkable combinatorial structure of matching allows for a very efficient algorithm.

\definition{(Tenacity of vertices and edges)}
\label{ref.tenacity}
Define the tenacity of vertex $v$, $\t(v) = \e(v) + \o(v)$.
If $(u, v)$ is an unmatched edge, its tenacity, $\t(u, v) = \e(u) + \e(v) + 1$, and if it is matched,
$\t(u, v) = \o(u) + \o(v) + 1$.

In Figure \ref{fig.vten}, the tenacity of each edge in the 5-cycle is 9 and the tenacity of the rest of the edges is $\infty$.
Figures \ref{fig.verten} and \ref{fig.edgeten} give the tenacity of vertices and edges, respectively, in a more interesting graph.


\begin{figure}[ht]
\begin{minipage}[b]{0.5\linewidth}
\centering
\includegraphics[width=\textwidth]{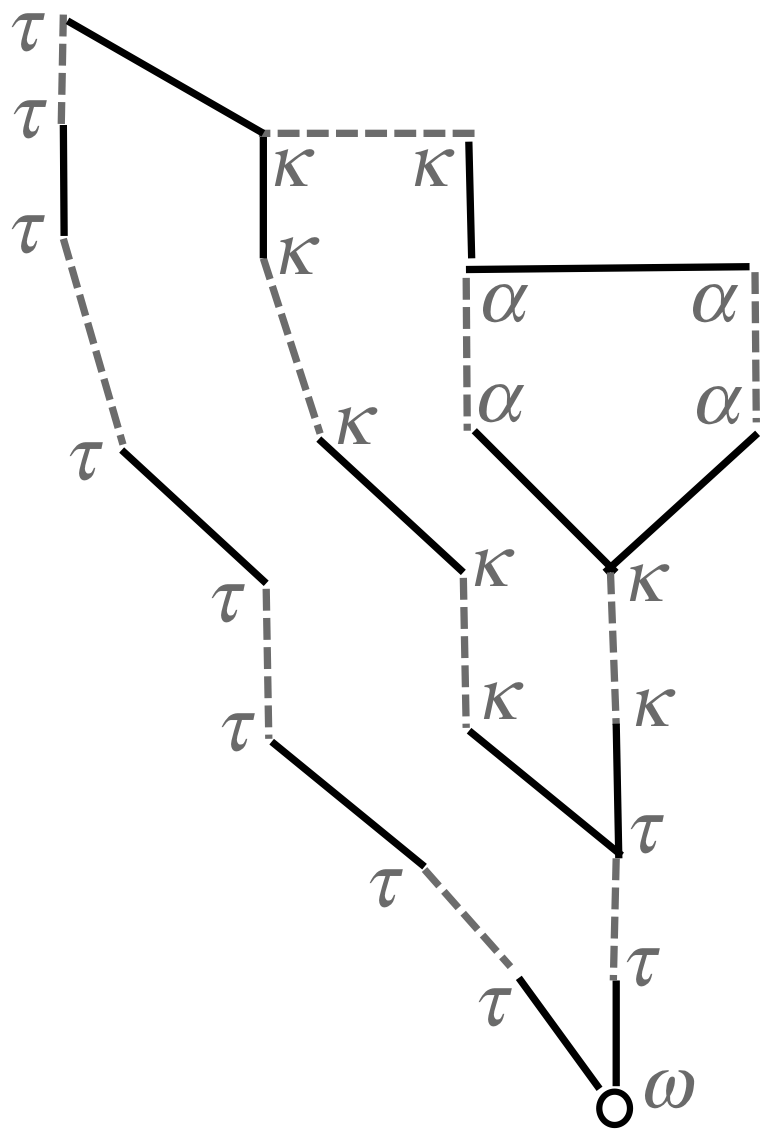}
\caption{The tenacity of vertices is indicated; here $\alpha = 13, \ \kappa = 15$, $\tau = 17$ and $\omega = \infty$.}
\label{fig.verten}
\end{minipage}
\hspace{0.5cm}
\begin{minipage}[b]{0.5\linewidth}
\centering
\includegraphics[width=\textwidth]{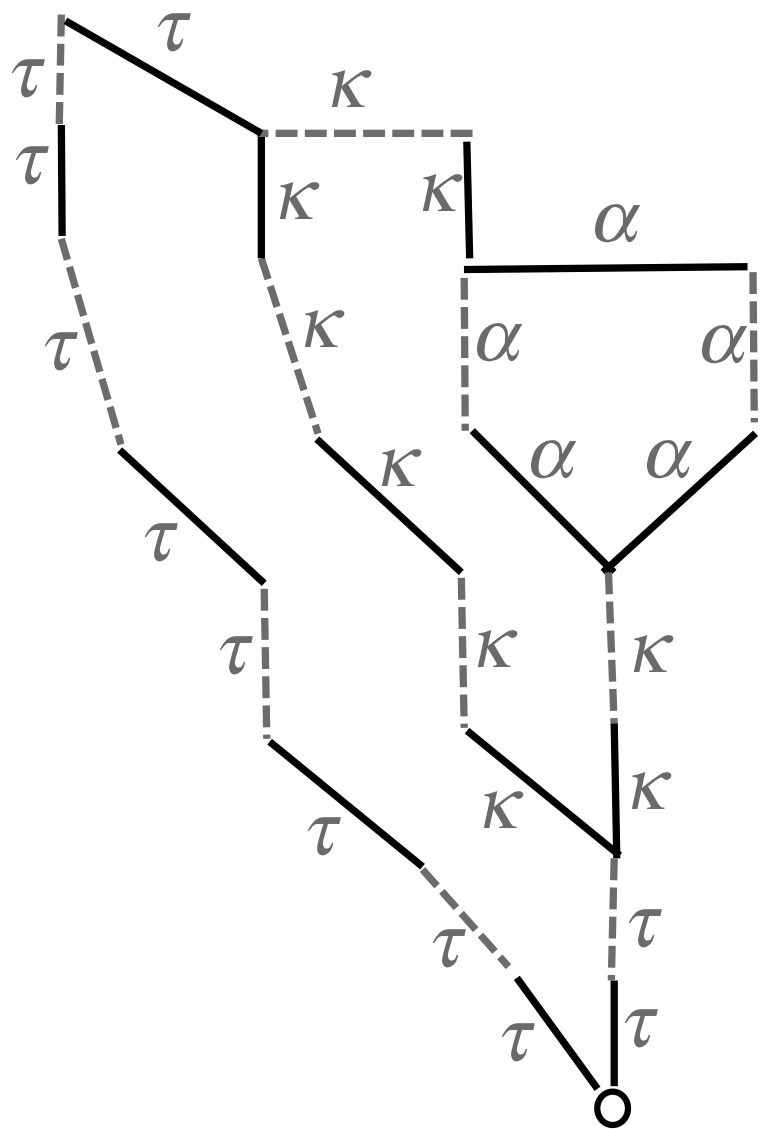}
\caption{The tenacity of each edge is indicated; here $\alpha = 13, \ \kappa = 15$ and $\tau = 17$.}
\label{fig.edgeten}
\end{minipage}
\end{figure}


\begin{lemma}
\label{lem.t-matched}
If $(u, v)$ is a matched edge, then $\t(u) = \t(v) = \t(u, v)$.
\end{lemma}

\begin{proof}
If $(u, v)$ is a matched edge, $\e(v) = \o(u) + 1$ and $\e(u) = \o(v) + 1$. The lemma follows.
\end{proof}

As a result of Lemma \ref{lem.t-matched}, in several proofs below it will suffice to restrict attention to only 
one of the end points of a matched edge.

\definition{(Limited BFS-honesty)}
Let $p$ be an $\e(v)$ or $\o(v)$ path starting at unmatched vertex $f$ and let $u$ lie on $p$.
Then $|p[f \ \mbox{to} \ u]|$ will denote the length of this path from $f$ to $u$, and if it is even (odd) we will say
that $u$ is {\em even (odd) w.r.t. p}. 
We will say that $u$ is BFS-honest w.r.t. $p$ if 
$|p[f \ \mbox{to} \ u]| = \e(u) \ (\o(u))$ if $u$ is even (odd) w.r.t. $p$.

The usefulness of tenacity is established in Theorem \ref{thm.honest}, which shows limited BFS-honesty even in the
non-bipartite case. 


Observe that in the graph of Figures \ref{fig.BFSH} and \ref{fig.verten}, the vertices $a, \ b$ and $c$ are BFS-honest on
all evenlevel and oddlevel paths to the vertices of tenacity $\alpha$. (However, the vertices of tenacity $\alpha$ are not
BFS-honest on $\o(a)$ and $\o(b)$ paths.)

\begin{theorem}
\label{thm.honest}
Let $p$ be an $\e(v)$ or $\o(v)$ path and let vertex $u \in p$ with $\t(u) \geq \t(v)$. Then $u$ is BFS-honest w.r.t. $p$.
Furthermore, if $\t(u) > \t(v)$ then $|p[f \ \mbox{to} \ u]| = \mn(u)$.
\end{theorem}

\begin{proof}
Assume w.l.o.g that $p$ is an $\e(v)$ path and that $u$ is even w.r.t. $p$ (by Lemma \ref{lem.t-matched}).
Suppose $u$ is not BFS-honest w.r.t. $p$, and let $q$ be an $\e(u)$ path, i.e., $|q| < |p[f \ \mbox{to} \ u]|$.
First consider the case that $\e(v) = \mx(v)$, and let $r$ be a $\mn(v)$ path. Let $u'$ be the matched neighbor of $u$.
Consider the first vertex of $r$ that lies on $p[u' \ \mbox{to} \ v]$. 
If this vertex is even w.r.t. $p$ then $\o(u) \leq |r| + |p[u \ \mbox{to} \ v]|$. Additionally, 
$\e(u) < |p[f \ \mbox{to} \ u]|$, hence $\t(u) < \t(v)$, leading to a contradiction. On the other hand, if this vertex is odd 
w.r.t. $p$ then $\mn(v) = |r| > \e(u)$, because otherwise there is a shorter even path from $f$ to $v$ than $\e(v)$. 
We combine the remaining argument along with the case that $\e(v) = \mn(v)$ below.

Consider the first vertex, say $w$, of $q$ that lies on $p(u \ \mbox{to} \ v]$ -- there must be such a vertex because otherwise
there is a shorter even path from $f$ to $v$ than $\e(v)$. If $w$ is odd w.r.t. $p$ then we get an even path to $v$
that is shorter than $\e(v)$. Hence $w$ must be even w.r.t. $p$. Then, $q[f \ \mbox{to} \ w] \circ p[w \ \mbox{to} \ u]$
is an odd path to $u$ with length less than $\e(v)$. Again we get $\t(u) < \t(v)$, leading to a contradiction. 
\end{proof}

\section{The base of a vertex}
\label{sec.base}

Let $v$ be vertex of tenacity $t < l_m$ and $p$ be an $\e(v)$ or $\o(v)$ path starting at unmatched vertex $f$. 
Define {\em the base of v w.r.t. p}, denoted $F(p, v)$, to be the vertex of tenacity $> t$ that is furtherest away from $f$ on $p$.
The main fact we will prove in this section is:

\begin{theorem}
\label{thm.base}
Let $v$ be vertex of tenacity $t < l_m$.
Then the set 
\[ B = \{F(p, v) ~|~ p \ \mbox{is an} \ \e(v) \ \mbox{or} \ \o(v) \ \mbox{path} \} \] 
is a singleton.
\end{theorem}

We will need some definitions to prove a preliminary fact first. Let $p$ be an alternating path starting at unmatched vertex $f$.
Given two vertices $u$ and $w$ on $p$, if $u$ is further away from $f$ on $p$ than $w$ then we will say that $u$ is {\em higher} than $w$.
An alternating path that starts and ends at matched vertices, say $u$ and $v$, of $p$ and does not intersect any vertices of $p$ 
other than $u$ and $v$ is said to be a {\em segment on p}; 
clearly, a segment will be of odd length. Let $s$ be a segment on $p$ whose endpoints are vertices $u$ and $v$ on $p$. Then we will say that all 
vertices on $p[u \ \mbox{to} \ v]$ are {\em covered by $s$}. Let $S$ be a set of segments on $p$ that are vertex disjoint.
A vertex $v$ on $p$ is {\em covered by $S$} if it is covered by one of the segments in $S$. The set $S$ is said to form a {\em flower on $p$} if 
it satisfies the following recursive conditions:

\begin{enumerate}
\item
$S$ consists of a single segment, $s$, which starts and ends at even vertices of $p$.
\item
$S = S' \cup \{ s \}$, where $S'$ forms a flower on $p$ and one endpoint of segment $s$ is covered by $S'$ and the other is even w.r.t. $p$.
\item
$S = S' \cup S'' \cup \{ s \}$, where $S'$ and $S''$ form flowers on $p$ and one endpoint of segment $s$ is covered by $S'$ and the other 
by $S''$.
\end{enumerate}

In this section, we will denote the matched neighbor of matched vertex $u$ by $u'$; furthermore, if the edge $(u, u')$ lies on $p$,
we will assume that $u$ is even w.r.t. $p$ and $u'$ is odd w.r.t. $p$.
Let $S$ be a flower on $p$ and let $u$ and $w$ be the lowest and highest vertices of $p$ covered by $S$. Then $u$ and $w$ will be called the 
{\em base} and {\em tip} of the flower, respectively. Observe that $u$ and $w$ will both be even w.r.t. $p$.
The {\em length} of this flower will be the sum of lengths of all segments in 
$S$ plus $|p[u \ \mbox{to} \ w]|$. The proof of the next lemma follows via an easy induction based on the above-stated recursive
definition of flower:

\begin{lemma}
\label{lem.flower}
Let $S$ form a flower on $p$ with base $u$, and let $w$ be any vertex covered by $S$. 
Then there are odd and even alternating paths from $u$ to $w$, each of length at most the length of the flower.
\end{lemma}

Let $q$ be an alternating path that intersects $p$, possibly at several places. Each odd length alternating subpath of $q$ that
starts and ends at $p$ will be called a {\em segment of $q$ on $p$}. 
The proof of the following lemma is straightforward.

\begin{lemma}
\label{lem.covered}
Let $q$ be an alternating path that intersects $p$ such that its first segment starts at even vertex $u$ and last segment ends
at even vertex $w$. Then at least one of $u$ and $w$ is covered by a flower formed by the segments of $q$ on $p$.
\end{lemma}

Let $v$ be a vertex of tenacity $t < l_m$ and let $p$ and $q$ be $\e(v)$ and $\o(v)$ paths, respectively, starting at unmatched vertex $f$. 
Consider vertices of tenacity $> t$ that appear on both $p$ and $q$; by Lemma \ref{thm.honest}, each such vertex must be BFS-honest w.r.t. 
both $p$ and $q$. Since $t < l_m$, $f$ is one such vertex\footnote{This is precisely the reason for assuming $t < l_m$ in this lemma and 
beyond; see also the Remark after the proof of Theorem \ref{thm.base}.}. 
Among these vertices, let $b$ be the highest. We remark that it is straightforward to show that neither 
$p(b \ \mbox{to} \ v]$ nor $q(b \ \mbox{to} \ v]$ can intersect $p[f \ \mbox{to} \ b]$ -- otherwise either $b$ will have tenacity 
at most $t$ or there will be a shorter path to $v$.

A matched edge $(w, w')$ that lies on both $p(b \ \mbox{to} \ v]$ and $q(b \ \mbox{to} \ v]$ is said to be a {\em common edge}.
If both these paths traverse this edge in the same direction then we will say that $(w, w')$ is a {\em forward edge} and otherwise
it is a {\em backward edge}. A common edge $(w, w')$ is said to be a {\em separator} if the graph consisting of the vertices and edges of 
$p[b \ \mbox{to} \ v] \cup q[b \ \mbox{to} \ v]$ gets disconnected by the removal of this edge; clearly, such an edge must be a forward edge.

We next define three types of frontier edges. Let $(w, w')$ be a separator edge and let $(w, u')$ and $(u, u')$ be the adjacent unmatched
and matched edges, respectively, higher up on $p$. If $(u, u')$ is not a separator edge then $(w, w')$ is called a {\em separator frontier edge}.
Next let $(w, w')$ be a matched edge that is not a separator edge and such that $q(f \ \mbox{to} \ w)$ has no intersection with $p(w \ \mbox{to} \ v)$. 
Then if $(w, w')$ is a forward (backward) edge, it will be called a {\em forward frontier (backward frontier)}. 
Figures \ref{fig.back1}, \ref{fig.forward1} and \ref{fig.separator1}, show backward, forward and separator frontiers, respectively,
and Figures \ref{fig.back2}, \ref{fig.forward2} and \ref{fig.separator2} show these three graphs redrawn in a manner that
reveals their structure more easily. In Figure \ref{fig.separator1}, the $\e(v)$ and $\o(v)$ paths both use edges $(b, w')$
and $(w', w)$ and the latter edge is a separator frontiers. The reason for drawing edge $(b, w)$ is that it provides an $\o(w)$ path.

The difference between a separator edge and a forward frontier edge is the following: in the former case, $p(b \ \mbox{to} \ v]$
does not intersect $p[f \ \mbox{to} \ w)$ and in the latter case it does. In both cases, $|p[f \ \mbox{to} \ w]| = |q[f \ \mbox{to} \ w]|$, 
since otherwise one of the paths to $v$ can be shortened.

\begin{figure}[ht]
\begin{minipage}[b]{0.5\linewidth}
\centering
\includegraphics[width=\textwidth]{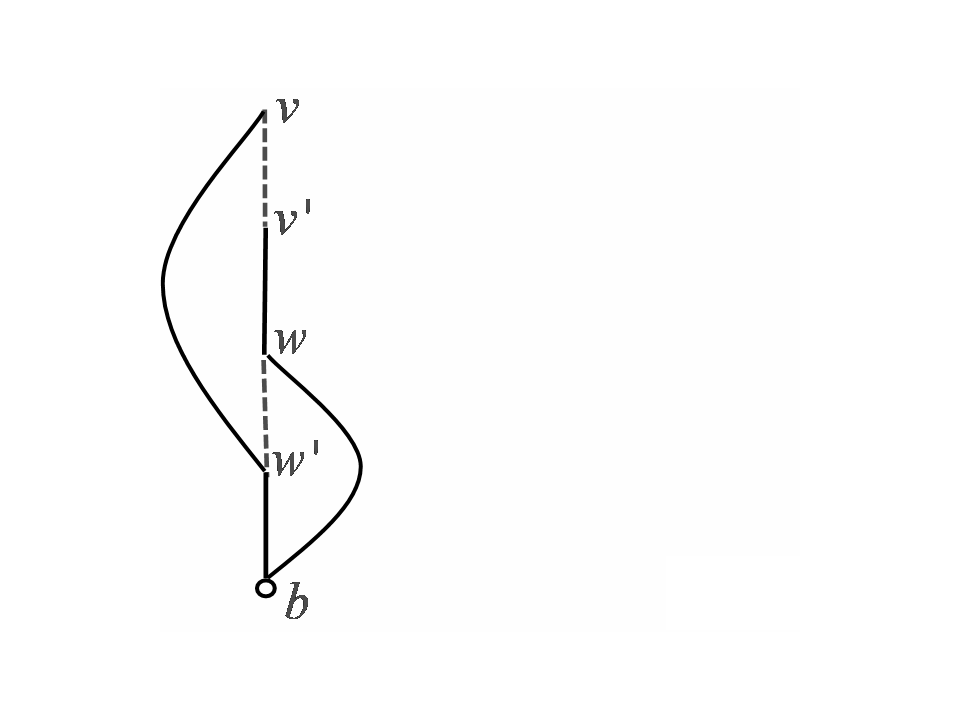}
\caption{
Edge $(w, w')$ is a backward frontier.}
\label{fig.back1}
\end{minipage}
\hspace{0.5cm}
\begin{minipage}[b]{0.5\linewidth}
\centering
\includegraphics[width=\textwidth]{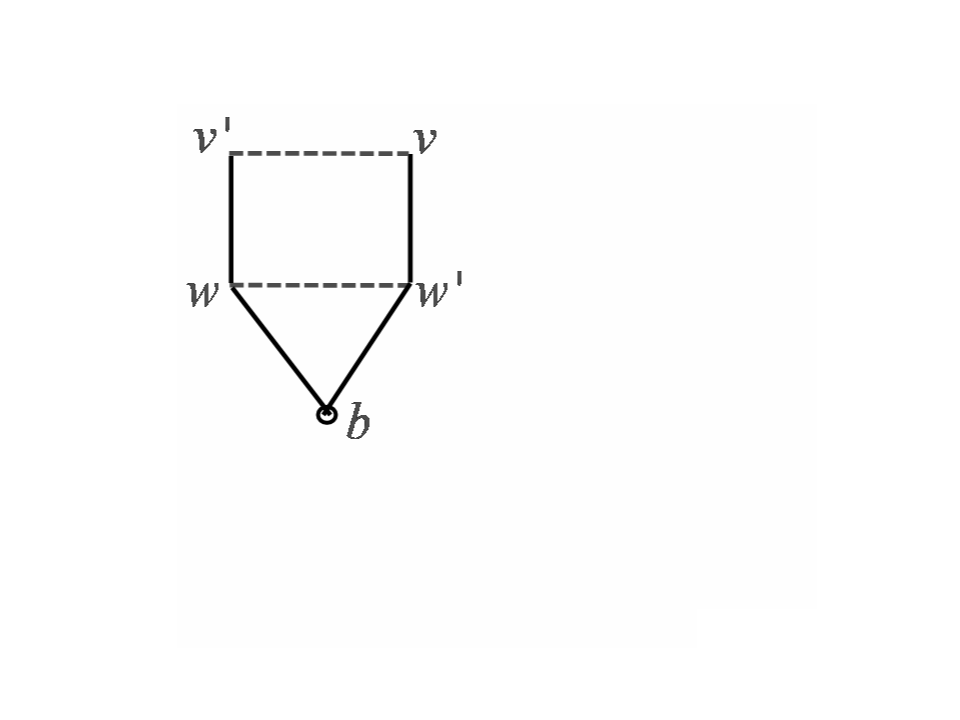}
\caption{}
\label{fig.back2}
\end{minipage}
\end{figure}

\begin{lemma}
\label{lem.pq}
Let $v$ be a vertex of tenacity $t < l_m$ and let $p$ and $q$ be $\e(v)$ and $\o(v)$ paths, respectively, starting at unmatched vertex $f$,
such that there are no separators on $p(b \ \mbox{to} \ v]$ and $q(b \ \mbox{to} \ v]$.
Then $F(p, v) = F(q, v)$.
\end{lemma}

\begin{proof}
We will prove that all vertices on $p(b \ \mbox{to} \ v]$
and $q(b \ \mbox{to} \ v]$ are of tenacity at most $t$, thereby proving the lemma. By Lemma \ref{lem.t-matched}, it suffices to 
prove this for vertices on $p(b \ \mbox{to} \ v]$ that are even w.r.t. this path. For each such vertex, say $u$, we will prove that
$\o(u) \leq |q| + |p[u \ \mbox{to} \ v]|$ thereby proving the claim, since $e(u) \leq |p[f \ \mbox{to} \ u]|$. 

Let $(w, w')$ be the frontier edge that is closest to $(u, u')$, such that $w$ is same as or above $u$. 
First consider the case that $(w, w')$ is a backward frontier. If so, $q[f \ \mbox{to} \ w] \circ p[w \ \mbox{to} \ u]$ is an
odd alternating path which obviously satisfies the length constraint.

Second consider the case that $(w, w')$ is a forward frontier. Now, $|q[f \ \mbox{to} \ w]| \geq |p[f \ \mbox{to} \ w]|$, since 
otherwise we can get a shorter even path to $v$ by splicing the first half of $q$ with the second half of $p$ at $w$. 
Since $(w, w')$ is not a separator, $q[w \ \mbox{to} \ v]$ must traverse an edge of $p$ below $(w, w')$ before reaching $v$. 
Let $(x, x')$ be the last edge below $(w, w')$ that $q[w \ \mbox{to} \ v]$ traverses. Then $(x, x')$ must be a backward edge, 
since otherwise $p[f \ \mbox{to} \ x] \circ q[x \ \mbox{to} \ v]$ will be a shorter odd path to $v$. 
Now, by Lemma \ref{lem.covered}, $x$ or $w$ will be covered by a flower formed 
by the segments of $q[w \ \mbox{to} \ v]$. In the former case, again one can get a shorter odd path to $v$, leading to a contradiction.
Therefore, $w$ is covered by a flower formed by the segments of $q[w \ \mbox{to} \ v]$. Let $b'$ be the base of this flower.

Now, consider vertex $u$ for which we need to place an upper bound on $\o(u)$. If $u$ lies below $x$ or if $u$ is not
covered by a flower formed by the segments of $q[w \ \mbox{to} \ v]$, then $q[f \ \mbox{to} \ w]$
concatenated with an odd path from $w$ to $b'$, through the flower covering $w$, concatenated with $p[b' \ \mbox{to} \ u]$
is an odd path to $u$ of the required length. Next assume that $u$ is covered by a flower formed by the segments 
of $q[w \ \mbox{to} \ v]$, with the base of this flower being $b''$. Then $p[f \ \mbox{to} \ b'']$ concatenated with an odd path from 
$b''$ to $u$ through this flower is the required odd path; its length is bounded by 
\[ |p| + |q[w \ \mbox{to} \ v]| \leq |q| + |p[u \ \mbox{to} \ v]|, \]
since $|q[f \ \mbox{to} \ w]| \geq |p[f \ \mbox{to} \ w]|$ and we did not use any part of $q[f \ \mbox{to} \ w]$. 
\end{proof}

\begin{figure}[ht]
\begin{minipage}[b]{0.5\linewidth}
\centering
\includegraphics[width=\textwidth]{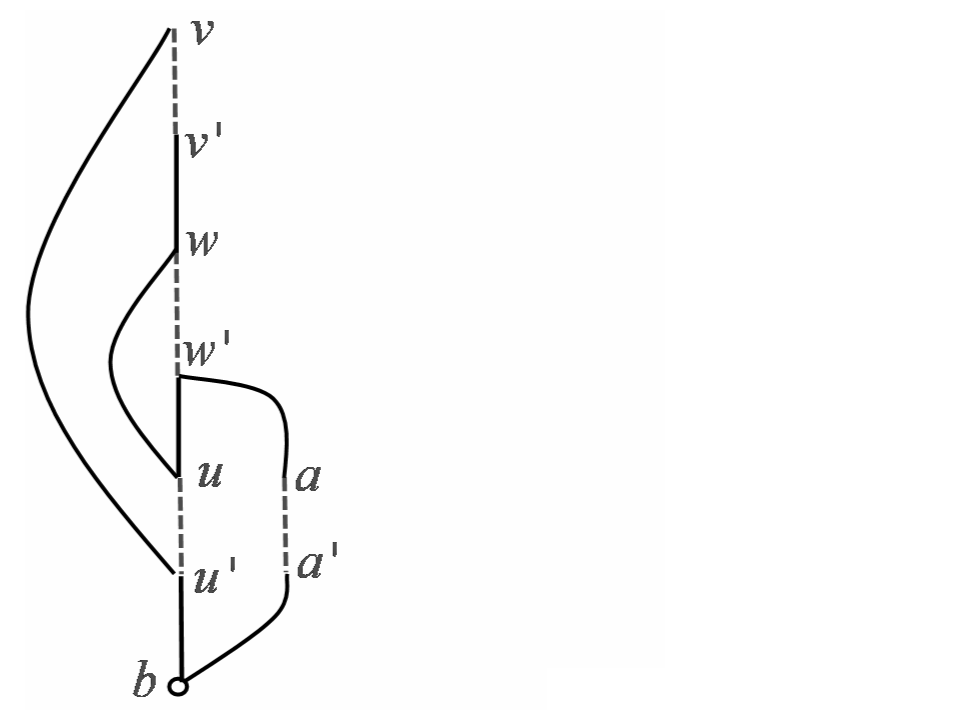}
\caption{
Edge $(w, w')$ is a forward frontier.}
\label{fig.forward1}
\end{minipage}
\hspace{0.5cm}
\begin{minipage}[b]{0.5\linewidth}
\centering
\includegraphics[width=\textwidth]{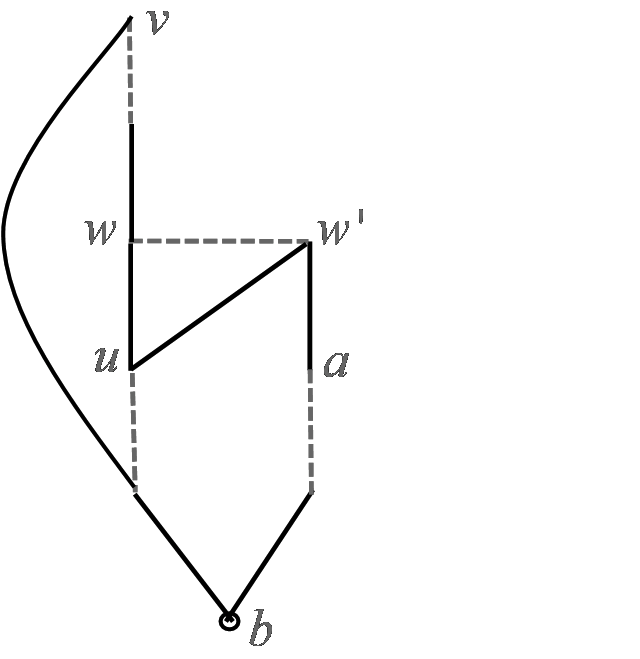}
\caption{}
\label{fig.forward2}
\end{minipage}
\end{figure}

\begin{lemma}
\label{lem.separator}
Let $v$ be a vertex of tenacity $t < l_m$ and let $p$ and $q$ be $\e(v)$ and $\o(v)$ paths, respectively, starting at unmatched vertex $f$.
Then $F(p, v) = F(q, v)$.
\end{lemma}

\begin{proof}
Lemma \ref{lem.pq} dealt with the case that $p$ and $q$ have no separators. Next assume that they do and let $(w, w')$ be
the lowest separator frontier edge. We claim that $w$ and $w'$ must be BFS-honest on $p$ and $q$ -- otherwise any $\e(w)$ path
must intersect $p(w \ \mbox{to} \ v]$ and $q(w \ \mbox{to} \ v]$ and one can then show that there must be a shorter even or odd
path to $v$ than $p$ or $q$, respectively. 

Next let us argue that all vertices on $p(b \ \mbox{to} \ w]$
and $q(b \ \mbox{to} \ w]$ are of tenacity at most $t$. As in Lemma \ref{lem.pq}, it suffices to restrict attention to
vertices on $p(b \ \mbox{to} \ w]$ that are even w.r.t. this path. For each such vertex, say $u$, we will prove that
$\o(u) \leq |q| + |p[u \ \mbox{to} \ v]|$. 

Since $w$ lies on both $p$ and $q$, and by the definition of vertex $b$, $\t(w) \leq t$. Let $r$ be an $\o(w)$ path; w.l.o.g. $r$
starts at unmatched vertex $f$. Since $w$ is BFS-honest w.r.t. $p$, $|r| = \t(w) - |p[f \ \mbox{to} \ w]|$. Let $h$ be the last vertex 
of $r[f  \ \mbox{to} \ w)$ that also lies on $p[f \ \mbox{to} \ w)$. First consider the easy case that $h = b$. If so, 
$r \circ p[w \ \mbox{to} \ u]$ is an odd path to $u$ of the required length.  

Next assume that $h \neq b$ and $h$ is even w.r.t. $p$. Now, for $u$ on $p(h \ \mbox{to} \ w)$, 
$r \circ p[w \ \mbox{to} \ u]$ is an odd path to $u$ of the required length. And for $u$ on $p(b \ \mbox{to} \ h]$, 
$q[f \ \mbox{to} \ w] \circ  r[w \ \mbox{to} \ h] \circ p[h \ \mbox{to} \ u]$ is the required path.
Finally, assume that $h$ is odd w.r.t. $p$. Then, for $u$ on $p(b \ \mbox{to} \ h]$, $r \circ p[w \ \mbox{to} \ u]$ is the required path.
And for $u$ on $p(h \ \mbox{to} \ w)$, $r \circ p[w \ \mbox{to} \ u]$ is the required path.

The arguments given above carry over to even vertices on $p$ that lie between any two separator frontier edges. 
The part of $p$ between the highest separator frontier edge and $v$ is similar to the no separators case and can be
handled using the arguments given in Lemma \ref{lem.pq}.
\end{proof}

\suppress{

REMOVE LATER

}

\begin{figure}[ht]
\begin{minipage}[b]{0.5\linewidth}
\centering
\includegraphics[width=\textwidth]{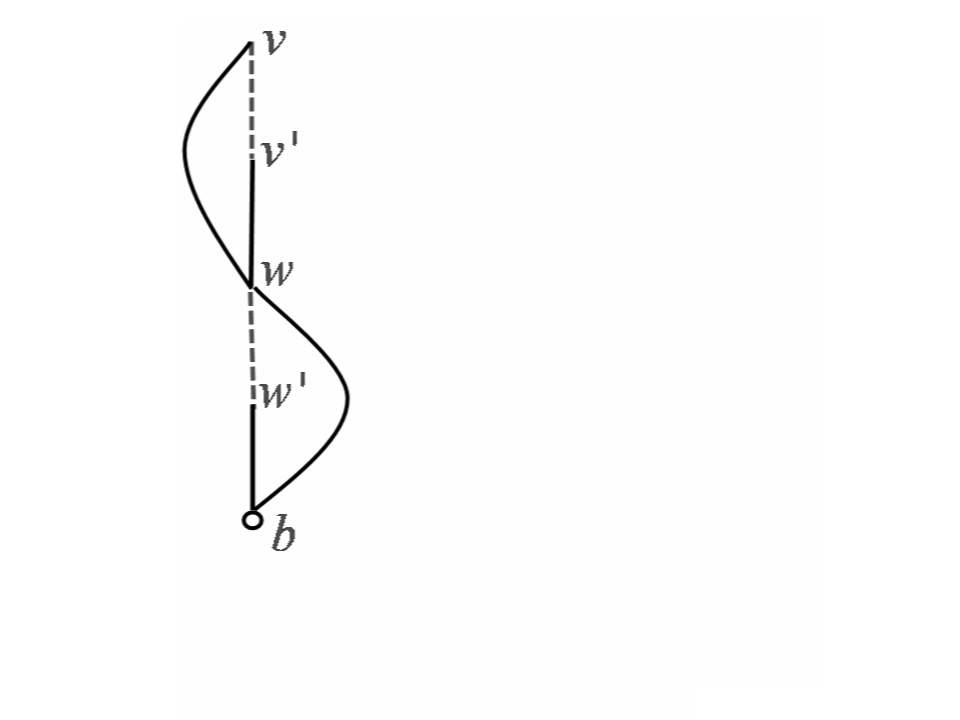}
\caption{
Edge $(w, w')$ is a separator frontier.}
\label{fig.separator1}
\end{minipage}
\hspace{0.5cm}
\begin{minipage}[b]{0.5\linewidth}
\centering
\includegraphics[width=\textwidth]{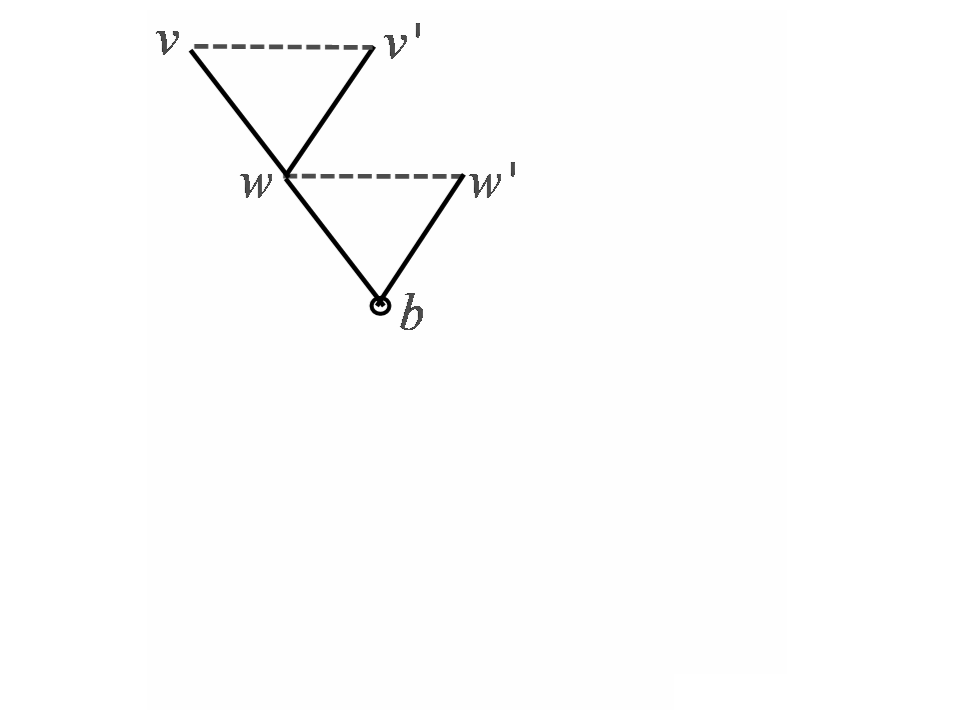}
\caption{}
\label{fig.separator2}
\end{minipage}
\end{figure}

\begin{proof}
{\bf (of Theorem \ref{thm.base})}
In Lemma \ref{lem.separator}, first fix $p$ and vary $q$ over all $\o(v)$ paths and then fix $q$ and vary $p$ over all $\e(v)$ paths
to prove that all $\e(v)$ and $\o(v)$ paths starting at the same unmatched vertex have the same base.
Next suppose that there are $\e(v)$ and $\o(v)$ paths starting from more than one unmatched vertex. 

Let $p$ and $q$ be $\e(v)$ paths starting from unmatched vertices $f$ and $f'$, respectively. Since $t < l_m$, 
these paths must meet at a vertex that is odd w.r.t. both paths, say $w'$. Now $p[w' \ \mbox{to} \ v]$ cannot intersect
$q[f' \ \mbox{to} \ w')$, since otherwise we can get a shorter even path to $v$ or an augmenting path of length $< l_m$.
Therefore, $q[f' \ \mbox{to} \ w') \circ p[w' \ \mbox{to} \ v]$ is also an $\e(v)$ path. Let us name it $q'$. By the fact
established in the previous paragraph, $F(q', v) = F(q, v)$.

Suppose $F(p, v) \neq F(q', v)$; let $F(p, v) = b_1$ and $F(q', v) = b_2$. Clearly, $b_1$ lies on $p[f \ \mbox{to} \ w')$
and $b_2$ lies on $q'[f' \ \mbox{to} \ w')$. Therefore, $\t(w') \leq t$. It is easy to see that an $\e(w')$ path together
with $p[w' \ \mbox{to} \ f]$ or $q'[w' \ \mbox{to} \ f']$ will yield an augmenting path of length $< l_m$, leading to a contradiction. 
Therefore, $F(p, v) = F(q', v)$, hence proving the theorem.
\end{proof}

\begin{figure}[h]
\begin{center}
\includegraphics[scale = 0.4]{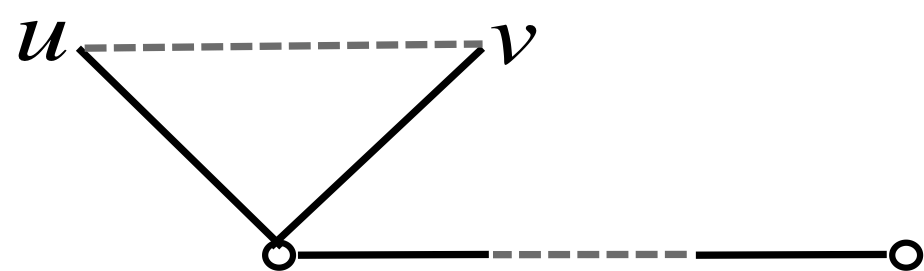}
\caption{Vertices $u$ and $v$ have no base.}
\label{fig.nobase}
\end{center}
\end{figure}

\Remark The condition ``$t < l_m$'' is essential in Theorem \ref{thm.base}; Figure \ref{fig.nobase} gives a counter-example.
In this graph $l_m = 3$ and the vertices $u$ and $v$, both of tenacity 3, have no base.

\definition{(Base of a vertex, and iterated bases)}
For a vertex $v$ of tenacity $t < l_m$ define its base to be the singleton vertex in the set defined in Theorem \ref{thm.base}.
We will denote it as $\base(v)$. The iterated bases of $v$ are defined as follows:
Define $\base^1(v) = \base(v)$, and for $k \geq 1$, if $\t(\base^k(v)) < l_m$, define $\base^{k+1}(v) = \base(\base^k(v))$.
For convenience, we will define $\base^0 (v) = v$ even though it is not an iterated base of $v$.

\definition{(Outer and inner vertices)}
A vertex $v$ is said to be {\em outer} if $\e(v) < \o(v)$ and {\em inner} otherwise.

Clearly, the base of any vertex is outer and hence all iterated bases of a vertex are outer vertices. Clearly, $\base^0 (v) = v$ will
be inner if $v$ is inner; however, as stated above, $\base^0 (v)$ is not an iterated base of $v$.
In the graph of Figures \ref{fig.BFSH} and \ref{fig.verten}, the base of each vertex of tenacity $\alpha$ is $a$,
tenacity $\kappa$ is $b$, and tenacity $\tau$ is $f$, respectively. Additionally, $\base^1(v) = a, \base^2(v) = b$ and $\base^3(v) = f$. 
Theorems \ref{thm.honest} and \ref{thm.base} give:

\begin{corollary}
\label{cor.bases}
For any vertex $v$, its iterated bases occur on every $\e(v)$ and $\o(v)$ path in a BFS-honest manner.
\end{corollary}

\definition{(Shortest path from an iterated base to a vertex)}
Let $v$ be a vertex of tenacity $t < l_m$ and let $b$ be an iterated base of $v$, i.e., $b = \base^k(v)$ for $k \in \Zplus$. 
By an $\e(b; v)$ ($\o(b; v)$) path we mean a minimum even (odd) length alternating path from $b$ to $v$ that starts with an unmatched edge. 

\begin{lemma}
\label{lem.concat}
Let $v$ be a vertex of tenacity $t < l_m$ and let $b = \base(v)$. Then every $\e(v)$ ($\o(v)$) path consists of an $\e(b)$ path 
concatenated with an $\e(b; v)$ ($\o(b; v)$) path.
\end{lemma}

\begin{proof}
Let $p$ be an $\e(b)$ path starting at unmatched vertex $f$ and $q$ be an $\e(b; v)$ path. If their concatenation is longer than 
$\e(v)$ then $q$ must intersect $p$ below $b$. Let $(w, w')$ be the lowest matched edge of $p$ used by $q$, where $w'$ is even w.r.t. $p$. 
Now, using the same arguments as those in the proof of Theorem \ref{thm.honest},
one can show that if $w$ is odd w.r.t. $q$ then there is an even path from $f$ to $v$ that is shorter than $\e(v)$ and
if $w'$ is odd w.r.t. $q$ then there is a short enough odd path from $f$ to $b$ which gives $\t(b) \leq t$. 
\end{proof}

\section{Blossoms and their properties vis-\'{a}-vis \\
minimum length alternating paths}
\label{sec.blossom}

\definition{\label{def.blossom} (Blossoms and their nesting depth)}
Blossoms will be defined recursively.
Let $b$ be an outer vertex and $t$ be an odd number such that $\t(b) > t$ and $t < l_m $. We will denote the blossom of
tenacity $t$ and base $b$ by $\bt$. Define $\CB_{b, 1} = \emptyset$ and define its nesting depth, $N(\CB_{b, 1}) = 0$.
If $t \geq 3$, let $S_{b, t} = \{v ~|~ \t(v) = t \ \mbox{and} \ \base(v) = b \}$ and define
\[ \bt = S_{b, t} \cup \left(\bigcup_{v \in (S_{b, t} \cup \{b\}), \ v \ \mbox{outer}} \CB_{v, t-2} \right) . \]
Define the nesting depth of this blossom, 
\[ N(\bt) = 1 +  \left( \max_{v \in (S_t \cup \{b\}), \ v \ \mbox{outer}} N(\CB_{v, t-2})  \right)  \]
if $S_{b, t} \neq \emptyset$ and $N(\CB_{b, t-2})$ otherwise.

It is obvious from Definition \ref{def.blossom} that if $v \in \bt$ then $\t(v) \leq t$.
In the graph of Figures \ref{fig.BFSH} and \ref{fig.verten}, the blossom $\CB_{a, \alpha}$ consists of vertices of tenacity $\alpha$,
the blossom $\CB_{b, \kappa}$ consists of vertices of tenacity $\alpha$ and $\kappa$, and
the blossom $\CB_{f, \tau}$ consists of vertices of tenacity $\alpha , \kappa$ and $\tau$. The nesting depths of these blossoms 
are 1, 2 and 3, respectively. 

\begin{lemma}
\label{lem.bb}
Let $v \in \bt$. Then $\exists k$ such that $1 \leq k \leq N(\bt)$ and $b = \base^k(v)$. Furthermore, all the vertices
$\base(v),$ $\ldots , \base^{k-1} (v)$ belong to $\bt$.
\end{lemma}

\begin{proof}
By induction on the nesting depth of blossom $\bt$. If $N(\bt) = 1$, by definition, $b = \base(v)$. 
To prove the induction step, suppose $N(\bt) = l + 1$. Now, if $v \in S_{b, t}$, i.e., $\t(v) = t$, then $\base(v) = b$.
Otherwise, $\exists u \in S_{b, t}$ such that $v \in \CB_{u, t-2}$. Clearly $N(\CB_{u, t-2}) \leq l$ and $\base(u) = b$.
By the induction hypothesis, $\exists ! k$ such that $l \geq k \geq 1$ and $u = \base^k(v)$. Hence, $b = \base^{k+1}(v)$ and $k+1 \leq l+1$.
Finally, by the induction hypothesis, $\base(v),$ $\ldots , \base^{k-1} (v)$ belong to $\CB_{u, t-2}$.
Hence $\base(v),$ $\ldots , \base^{k} (v)$ belong to $\bt$.
\end{proof}

\begin{lemma}
\label{lem.bb1}
Let $t \leq t' < \t(b)$ and $t' < l_m$, and let
$\bt$ and $\CB_{b, t'}$ be two blossoms with the same base $b$. Then $\bt \subseteq \CB_{b, t'}$.
\end{lemma}

\begin{proof}
{\bf 1.} The proof is by induction on $t' - t$. The base case, i.e., $t = t'$, is obvious. Assume the induction hypothesis
that $\bt \subseteq \CB_{b, t' - 2}$. Now, by Definition \ref{def.blossom} it is straightforward that $\CB_{b, t' - 2} \subseteq \CB_{b, t'}$.
Hence $\bt \subseteq \CB_{b, t'}$.
\end{proof}

\begin{lemma}
\label{lem.vinb}
Let $\bt$ be a blossom with base $b$ and tenacity $t < l_m$, and $v$ be a vertex satisfying \; $\base^k (v) = b$ for some $k \geq 1$.
If $t \geq \t(base^{k-1} (v))$ then $v \in \bt$.
\end{lemma}

\begin{proof}
The proof is by induction on $k$. In the base case, i.e., $k = 1$, $\base(v) = b$. Let $\t(v) = r$, clearly $r \leq t$.
By Definition \ref{def.blossom}, $v \in \CB_{b, r}$ and by Lemma \ref{lem.bb1}, $\CB_{b, r} \subseteq \bt$.
Hence $v \in \bt$.

For the induction step, let $\base^{k-1} (v) = u$, $\t(u) = r \leq t$. By the induction hypothesis, 
$v \in \CB_{u, r-2}$. Since $\base(u) = b$, by Definition \ref{def.blossom}, $\CB_{u, r-2} \subseteq \CB_{b, r}$ and by Lemma \ref{lem.bb1},
$\CB_{b, r} \subseteq \bt$. Hence $v \in \bt$.
\end{proof}

\begin{lemma}
\label{lem.bb2}
Let $\bt$ and $\CB_{b', t'}$ be two blossoms such that $b \in \CB_{b', t'}$. Then $\bt \subset \CB_{b', t'}$.
\end{lemma}

\begin{proof}
By Lemma \ref{lem.bb}, there is a $k \geq 1$ such that $b' = \base^k(b)$ 
and $\base^1 (b), \ldots, \base^{k-1} (b) \in \CB_{b', t'}$. Clearly, $t' \geq \t(\base^{k-1} (b))$. To prove the statement,
we will apply induction on $k$. 

For the base case, i.e., $k = 1$, let $\t(b) = r$. Clearly, $t < r \leq t'$ and $\bt \subset \CB_{b, r-2}$.
By Definition \ref{def.blossom}, $\CB_{b, r-2} \subset \CB_{b', r}$, where the containment is proper since $b$ is not in the 
first blossom but it is in the second one. By Lemma \ref{lem.bb1}, $\CB_{b', r} \subseteq \CB_{b', t'}$.
Hence $\bt \subset \CB_{b', t'}$.

For the induction step assume $\base^{k+1} (b) = b'$.
Let $\base^{k} (b) = v$, and let $\t(v) = r$. Since $v \in \CB_{b', t'}$, $r \leq t'$. Clearly, $\t(\base^{k-1} (b)) \leq r-2$. 
Therefore, by Lemma \ref{lem.vinb}, $b \in \CB_{v, r-2}$. Furthermore, since $\base^{k} (b) = v$, by the induction hypothesis, 
$\bt \subset \CB_{v, r-2}$.

Since $\base(v) = b'$, by Definition \ref{def.blossom}, $\CB_{v, r-2} \subseteq \CB_{b', r}$. Since $r \leq t'$, 
$\CB_{b', r} \subseteq \CB_{b', t'}$. Hence, $\bt \subset \CB_{b', t'}$.
\end{proof}

\begin{theorem}
\label{thm.laminar}
The set of blossoms in $G$ forms a laminar family, i.e., two blossoms are either disjoint or one is contained in the other.
\end{theorem}

\begin{proof}
Suppose $v$ lies in blossoms $\bt$ and $\CB_{b', t'}$. If $b = b'$, we are done by the first claim in Lemma \ref{lem.bb1}.
Next assume that $b \neq b'$. Then by Lemma \ref{lem.bb2}, $b = \base^k(v)$ and $b' = \base^l(v)$; assume $k < l$.
By the second claim in Lemma \ref{lem.bb2}, $b = \base^k(v) \in \CB_{b', t'}$. Finally, by the second claim in Lemma \ref{lem.bb1}, $\bt \subset \CB_{b', t'}$.
\end{proof}

For the next theorem, we will need the following definition. Let $k = \min_{l} \{l \in \Zplus ~|~  \t(\base^l(v)) \geq t \}$ and let
$b = \base^k (v)$. Then we will say that {\em $b$ is the first iterated base of $v$ having tenacity at least $t$}.

\begin{theorem}
\label{thm.path}
Let $\t(v) = t < l_m$, $\base(v) = b$, and let $p$ be an $\e(b; v)$ or $\o(b; v)$ path. Then the following hold:
\begin{enumerate}
\item
Let vertex $u$ lie on $p$ with $u \neq b$ and let $b'$ be the first iterated base of $u$ having tenacity at least $t$.
If $b'$ lies on $p$ then $p[b' \ \mbox{to} \ u]$ is an $\e(b'; u)$ or $\o(b'; u)$ path, depending on the parity of $|p[b' \ \mbox{to} \ u]|$.
\item
Let $u$ and $b'$ satisfy the above-stated conditions. Then $b'$ lies on $p$.
\item
All vertices of $p$ are in $\bt \cup \{b\}$.
\item
Let vertex $u$ lie on $p$ with $u \neq b$. Then all iterated bases of $u$, until $b$, lie on $p$.   
\end{enumerate}
\end{theorem}

\begin{proof}
By Lemma \ref{lem.t-matched} we may assume w.l.o.g. that $p$ be an $\e(b; v)$ path.
We will prove the four claims by strong induction on $t$. For the base case, let $t$ be the smallest tenacity of a vertex in $G$.
In this case all vertices on $p$, other than $b$, will be of tenacity $t$ and will have base $b$. Hence all the claims are satisfied.
We prove the induction step below.

{\bf (1).} 
If $\t(u) = t$ then $b' = b$ and by Theorem \ref{thm.honest}, $u$ is BFS-honest w.r.t. $p$. Now the claim follows by Corollary \ref{cor.bases}
and Lemma \ref{lem.concat}.

Therefore if $u$ does not satisfy the claim, $t(u) < t$. Among the vertices that do not satisfy the claim let $u$ be the
last one on $p$. By Lemma \ref{lem.t-matched}, $u$ will be even w.r.t. $p$ and by the choice of $b'$, $u \in \CB_{b', t-2}$.
By the induction hypothesis applied to claim (3), every $\e(b'; u)$ path is contained in $\CB_{b', t-2} \cup \{b'\}$.

Let $w$ be any vertex on $p(u \ \mbox{to} \ v)$. Either $\t(w) = t$ and $\base(w) = b$ or by the choice of $u$, $w$ is
of tenacity less than $t$ and satisfies the condition of this claim. In the latter case, let $a$ be the first iterated base of $w$ 
having tenacity at least $t$. Then, $p[a \ \mbox{to} \ w]$ is contained in $\CB_{a, t-2} \cup \{a\}$.
By Theorem \ref{thm.laminar}, blossom $\CB_{a, t-2}$ is disjoint from $\CB_{b', t-2}$.  Therefore, any $\e(b'; u)$ path does not 
intersect $p(u \ \mbox{to} \ v)$.
Hence an $\e(b') \circ \e(b'; u) \circ p[u \ \mbox{to} \ v]$ path is a shorter path than $p$, leading to a contradiction.

{\bf (2).}
If $t(u) = t$ then $u$ is BFS-honest w.r.t. $p$ and by Corollary \ref{cor.bases}, $b'$ lies on $p$.
Therefore if $u$ does not satisfy the claim, $t(u) < t$. Among the vertices that do not satisfy the claim let $u$ be the last one on $p$.
By Lemma \ref{lem.t-matched}, $u$ will be even w.r.t. $p$ and by the induction hypothesis applied to claim (3), the path $\e(b'; u)$ is 
contained in $\CB_{b', t-2} \cup \{b'\}$. Now there are two cases:

{\bf Case 1:} 
The $\e(b')$ path does not intersect $p(u \ \mbox{to} \ v]$. \\
In this case, the path $\e(u) \circ p[u \ \mbox{to} \ v]$ is a shorter path than $p$,
leading to a contradiction.

{\bf Case 2:} 
The $\e(b')$ path intersects $p(u \ \mbox{to} \ v]$.  \\
Let $w$ be the first intersection of the $\e(b')$ path with $p(u \ \mbox{to} \ v]$. 
Now $w$ must be even w.r.t. $p$, since otherwise we can obtain a shorter even path to $v$.
Furthermore, $\t(w) = t$ since the $\e(b')$ path must enter any blossom of tenacity $t-2$ through its base.
Let $q$ be the path obtained by concatenating the part of the $\e(b')$ path from $w$ to $b'$ together with an $\o(b'; u)$ path.
Since every $\o(u)$ path consists of an $\e(b')$ path concatenated with an $\o(b'; u)$ path and $p[w \ \mbox{to} \ u]$ does not
go through $b'$, $p[w \ \mbox{to} \ u]$ must be longer than $q$. Hence by replacing $p[w \ \mbox{to} \ u]$ by $q$ in $p$
yields a shorter path than $p$, leading to a contradiction.


{\bf (3).}
By claim (2), if $u$ is on $p$ and $u \neq b$, then the first iterated base of $u$ having tenacity at least $t$, say $b'$, lies on $p$, and
by claim (1), $p[b' \ \mbox{to} \ u]$ is an $\e(b'; u)$ or $\o(b'; u)$ path, depending on the parity of $|p[b' \ \mbox{to} \ u]|$.
Therefore, one of the following three must hold:
$u = b$ or $\t(u) = t$ and $\base(u) = b$ or $u \in \CB_{b', t-2}$ where $b' = b$ or $\t(b') = t$ and $\base(b') = b$.
In all cases, $u \in ( \bt \cup \{b\} )$. The claim follows.

{\bf (4).}
If $\t(u) = t$, the claim is obvious.
Otherwise, $\t(u) < t$ and by claim (2), $b'$ lies on $p$. 
Now by claim (1),  
$p[b' \ \mbox{to} \ u]$ is an $\e(b'; u)$ or $\o(b'; u)$ path, depending on the parity of $|p[b' \ \mbox{to} \ u]|$.
Hence, we are done by Corollary \ref{cor.bases} and the induction hypothesis.
\end{proof}


\Remark  It is interesting to note that perhaps the most elementary claim in Theorem \ref{thm.path} appears to be the following subclaim of (4): 
If vertex $u$ lies on $p$, with $u \neq b$, then $\base(u)$ lies on $p$. However, all our attempts at proving this fact first failed.

Consider the $\o(a)$ and $\e(c)$ paths in Figure \ref{fig.BFSH}. Vertex $v$ lies on both and is not BSF-honest w.r.t. either path;
the iterated bases of $v$ are $a$, $b$ and $f$.
Now, $\base(a) = b$ and $\base(c) = f$. Let $p$ be an $\o(b; a)$ path and $q$ be an $\e(f; c)$ path. Clearly, $p \in \CB_{b, 15} \cup \{ b \}$
and $q \in \CB_{f, 17} \cup \{ f \}$. The relevant iterated bases of $v$, namely $a$ and $b$ lie on $p$ and $a, b$ and $f$ lie on $q$.
Finally, $p[a \ \mbox{to} \ v]$ is an $\e(a; v)$ path and $q[b \ \mbox{to} \ v]$ is an $\e(b; v)$ path despite the fact that
$v$ is not BSF-honest w.r.t. both $p$ and $q$.

Theorem \ref{thm.path} leads to the startling fact proved in Theorem \ref{thm.bases}; this fact also underlines the central importance
of the notion of base of a vertex. We first need some definitions. Let us say that vertex $v$ {\em has all iterated bases defined} if
for some $k \in \Zplus$, $base^k (v)$ is an unmatched vertex, and vertex $v$ is {\em useful} if it lies on a minimum length augmenting path.
It is easy to see that any useful vertex has all iterated bases defined; however, the reverse may not hold. Clearly, in both cases,
$\t(v) < l_m$.

\begin{theorem}
\label{thm.bases}
Let $v$ be a vertex that has all iterated bases defined and let $p$ be an $\e(v)$ or $\o(v)$ path. Then for each vertex $u$ that lies on $p$,
all iterated bases of $u$ lie on $p$.
\end{theorem}

\begin{proof}
First, it is easy to see that vertex $u$ has all iterated bases defined. If $u$ lies on $p(\base(v) \ \mbox{to} \ v]$, then by claim (4)
in Theorem \ref{thm.path} and Corollary \ref{cor.bases}, the theorem is true for $u$.

Assume that $base^k (v)$ is an unmatched vertex, where $k \in \Zplus$. If $u = base^k (v)$, the claim is clearly true.
Next assume that $u$ lies on \\
$p(\base^{l+1}(v) \ \mbox{to} \ \base^l (v)]$, where $1 \leq l < k$. Since $\base^l (v)$ is BFS-honest w.r.t. $p$, \\
$p[\base^{l+1}(v) \ \mbox{to} \ \base^l (v)]$ is an 
$\e(\base^{l+1}(v); \base^l (v))$ path, and again we are done by claim (4) in Theorem \ref{thm.path} and Corollary \ref{cor.bases}.
\end{proof}

As an illustration of Theorem \ref{thm.bases}, consider the $\e(c)$ path, say $p$, in the graph of Figure \ref{fig.BFSH}.
Vertex $v$ lies on $p$ and so do $a = \base(v), \ b = \base^2(v)$ and $f = \base^3(v)$.

\Remark The example stated above raises the following question: Let $v$ be a vertex that has all iterated bases defined and let
$p$ be an $\e(v)$ or an $\o(v)$ path starting at unmatched vertex $f$. If $u$ lies on $p$ then in what order do the iterated bases of 
$u$ occur on $p$? 

Let $f = \base^k(u)$. Clearly, if $u$ is BFS-honest on $p$, then 
$p[f \ \mbox{to} \ u]$ is an $e(u)$ or $\o(u)$ path and therefore the iterated bases of $u$ occur in the order $\base^k(u),
\base^{k-1}(u), \ldots, \base(u)$ on $p$. Next assume that $u$ is not BFS-honest on $p$ and that $u \in \bt$, where $t = \t(v)$ and
$b = \base(v)$. If not, then there is a vertex $w$ that lies on $p(u \ \mbox{to} \ v)$ such that $w$ is BFS-honset w.r.t. $p$ and
$u \in \CB_{b', t'}$, where $t' = \t(w)$ and $b' = \base(w)$; now, replace $v$ by $w$ in the remaining discussion.
By Theorem \ref{thm.path}, $\base^l(u) = b$ for $l \leq k$ and
$\base^{l+j}(u) = \base^{j-1}(v)$, for $1 \leq j \leq k-l$. Therefore,  $\base^k(u), \base^{k-1}(u), \ldots, \base^{l}(u)$ will occur in 
this order on $p$. The order of the remaining bases of $u$ will depend on the nesting structure of the sub-blossoms of $\bt$ ($\CB_{b', t'}$) 
and the manner in which $p$ traverses them.

\section{Bridges and their support}
\label{sec.bridge}

\definition{(Predecessor, prop and bridge)}
Consider a $\mn(v)$ path and let $(u, v)$ be the last edge on it; clearly, $(u, v)$ is matched if $v$ is outer and unmatched otherwise.
In either case, we will say that $u$ is a predecessor of $v$ and that edge $(u, v)$ is a prop.
An edge that is not a prop will be defined to be a bridge.

In Figure \ref{fig.BFSH}, the two horizontal edges and the oblique unmatched edge at the top 
are bridges and the rest of the edges of this graph are props.

\begin{theorem}
\label{thm.bridge}
Let $v$ be a vertex of tenacity, $t \leq l_m$, and let $p$ be a $\mx(v)$ path. Then there exists a unique bridge of tenacity $t$ on $p$.
\end{theorem}

\begin{proof}
Let $p$ start at unmatched vertex $f$. If $t < l_m$, let $\base(v) = b$. By Lemma \ref{lem.concat}, $p$ consists of an $\e(b)$ path concatenated 
with an $\e(b; v)$ path. Denote the latter by $q$. If $t = l_m$, let $q = p$.

By Theorem \ref{thm.honest}, each vertex $u$ of tenacity $t$ on $q$ is BFS-honest w.r.t. $p$. 
Let us partition these vertices into two sets: $S_1$ ($S_2$) consists of vertices $u$ such that $|p[f \ \mbox{to} \ u]| = \mn(u)$
($ = \mx(u)$). Let $S_1^{'} = S_1 \cup \{ b \}$. Clearly $v \in S_2$. Hence $S_1^{'}$ and $S_2$ are both non-empty. Let $a$ be the vertex 
in $S_1^{'}$ having the largest minlevel and $c$ be the vertex in $S_2$ having the smallest maxlevel. Now there are two cases.

{\bf Case 1:} $a$ and $c$ are adjacent on $p$ and $(a, c)$ is a matched edge. By Lemma \ref{lem.t-matched}, $\t(a) = \t(c) = \t(a, c)$.
Furthermore, for both $a$ and $c$, their minlevel is their oddlevel, therefore $(a, c)$ is not a prop. Hence it is a bridge of tenacity $t$.

{\bf Case 2:} In the remaining case, $a$ and $c$ must both be outer vertices, and in general, there may be several vertices of tenacity 
less than $t$ between $a$ and $c$ on $p$. By Theorem \ref{thm.path}, the first iterated base of these vertices having tenacity at least $t$ 
must be $a$ or $c$. Furthermore, by the third statement of Theorem \ref{thm.path}, the ones having base $a$ must be contiguous on $p$ and so 
must be the ones having base $c$. Let $a'$ be the last vertex on $p$ whose iterated base is $a$; if there is no such vertex, let $a' = a$.
Similarly, let $c'$ be the first vertex on $p$ whose iterated base is $c$; if there is no such vertex, let $c' = c$.
Then $(a', c')$ will be an unmatched edge of $p$. We will show that it is a bridge of tenacity $t$.

By Theorem \ref{thm.path}, $p[a \ \mbox{to} \ a']$ is an $\e(a; a')$ path and $p[c \ \mbox{to} \ c']$ is an $\e(c; c')$ path.
Now, $\t(a', c') = \e(a') + \e(c') + 1$. Substituting $\e(c) = t - \o(u)$, $\e(a') = \e(a) + \e(a; a')$ and $\e(c') = \e(c) + \e(c; c')$,
we get:
\[  \t(a', c') = \e(a) + \e(a; a') + t - \o(u) + \e(c; c') + 1 \ \ = \ \ t .\]
Clearly, $a$ gets its minlevel from its matched neighbor and if $a' \neq a$, $a'$ gets its minlevel from the blossom $\CB_{a, t-2}$.
Similarly, $c$ gets its minlevel from its matched neighbor and if $c' \neq c$, $c'$ gets its minlevel from the blossom $\CB_{c, t-2}$.
Therefore $(a', c')$ is not a prop. Hence it is a bridge of tenacity $t$.

Finally, we show that none of the remaining edges on $q$ is a bridge of tenacity $t$. Consider an edge $(i, j)$ on $p[b \ \mbox{to} \ a]$,
with $i$ below $j$ on $p$. If $\t(j) = t$ then $(i, j)$ must be a prop and if $t(j) < t$ then $j$ lies in a blossom nested in $\bt$
and $t(i, j) < t$. A similar argument holds for the edges on $p[c \ \mbox{to} \ v]$.
\end{proof}

\definition{(The support of a bridge)}
Let $(u, v)$ be a bridge of tenacity $t \leq l_m$. Then, its support is defined to be $\{w ~|~ \t(w) = t \ \mbox{and} \ \exists \ \mbox{a} \
\mx(w) \ \mbox{path containing} \ (u, v) \}$.

In the graph of Figures \ref{fig.BFSH}, \ref{fig.verten} and \ref{fig.edgeten}, the supports of the bridges of tenacity $\alpha, \kappa$
and $\tau$ are the set of vertices of tenacity $\alpha, \kappa$ and $\tau$, respectively.

\section{Double depth first search}
\label{sec.DDFS}

Double depth first search (DDFS) consists of
two coordinated depth first searches. The MV algorithm executes DDFS on the given graph, $G$; however, for ease of
exposition we will describe it in the simplified setting of a directed, layered graph $H$.

{\bf The layered graph:}
The vertices of $H$ are partitioned into $h + 1$ layers, $l_h, \ldots l_0$, with $l_h$ being the highest and $l_0$ the 
lowest layer. Each directed edge runs from a higher to a lower layer. Layer $l_h$ contains two vertices, $r$ and $g$, for
{\em red} and {\em green}. The edges of $H$ are such that from each vertex there is a path to a vertex in $l_0$. A vertex $v$ will be called
a {\em bottleneck} if every path from $r$ to $l_0$ and every path from $g$ to $l_0$ contains $v$. Observe that $v$ may be in layer $l_0$.
Among the bottlenecks, if there are any, the one having highest level will be called the {\em highest bottleneck}. If so, we will denote it 
by $b$. Let $V_b$ ($E_b$) be the set of all vertices (edges) that lie on all paths from $r$ and $g$ to $b$. 
If there are no bottlenecks, there must be distinct vertices $r_0$ and $g_0$ in layer $l_0$ and disjoint paths from $r$ to $r_0$ 
and $g$ to $g_0$.  If so, let $E_p$ be the set of all edges that lie on all paths from $r$ to $r_0$ and from $g$ to $g_0$.

{\bf The objective of DDFS:}
The purpose of DDFS is to find the highest bottleneck if one exists, and otherwise, to find distinct vertices $r_0$ and $g_0$ in layer $l_0$ 
and disjoint paths from $r$ to $r_0$ and $g$ to $g_0$. The running time of DDFS needs to be a function of the output in the following
manner: in the former case, DDFS needs to run in time $O(|E_b|)$ and in the latter case, it needs to run in time $O(|E_p|)$. 
In the former case, DDFS also needs to partition the vertices in $V_b - \{b\}$ into two sets $R$ and $G$, 
with $r \in R$ and $g \in G$. Furthermore DDFS needs to find two trees, $T_r$ and $T_g$, rooted at $r$ and $g$, such that the set of 
vertices visited by them is $R \cup \{b\}$ and $G \cup \{b\}$, respectively. 
The two trees are found by the two DFSs, called {\em red DFS} and {\em green DFS}, respectively. 

{\bf The two DFSs and their coordination:}
The two DFSs maintain their own stacks, $S_r$ and $S_g$, which start with $r$ and $g$, respectively. At any point
in the search, each stack contains the vertices that have been visited by the corresponding DFS but have not yet backtracked from. 
Each DFS performs the usual operations, with one important modification. The latter is described below when we give the rules for coordinating 
the two DFSs. Because edges in $H$ go from higher to lower levels, neither DFS will encounter back edges. The two DFSs start with their
center of activity at $r$ and $g$, respectively. Assume that the center of activity of a DFS is at $u$ and it searches along
edge $(u, v)$. If $v$ is not yet visited by either search, it pushes $v$ onto its stack and moves its center of activity to $v$.
In this case, $u$ is declared {\em parent of} $v$. If $v$ is already visited by either search, it considers the next unsearched edge incident 
at $u$ -- see below an exception, which is also the important modification mentioned above.
If all edges incident at $u$ have already been searched, it pops $u$ from its stack and moves the center of activity to the parent of $u$. 

We next give the rules for coordinating the two DFSs. To meet the running time requirement, we posit that if $b$ is the 
highest bottleneck in $H$, then neither DFS will search along any edges below $b$. This gives our first rule: as soon as 
the center of activities of the two DFSs are at different levels, the one that is higher must move to catch up. If both are at the
same level, we adopt the convention that red moves first. The exception mentioned above happens when one DFS searches along edge $(u, v)$ 
and $v$ happens to be the center of activity of the other DFS. In this case, $v$ could potentially be a bottleneck and the two DFSs first 
need to determine whether it is. Furthermore, if $v$ is not a bottleneck, the two DFSs need to determine which tree gets $v$.

{\bf When the two DFSs meet at a vertex:}
The procedure followed by the two DFSs at this point is the following, independent of which one got to $v$ first.
Let us assume that the red and green DFSs reached $v$ via edges $(v_r, v)$ and $(v_g, v)$, respectively.
First the green DFS backtracks from $v$ and tries to reach a vertex, say $w$, with $w \neq v$ and $\lv(w) \leq \lv(v)$. 
If green succeeds, red moves its center of activity to $v$ and it adds $v$ to $R$ and edge $(v_r, v)$ to $T_r$, and DDFS proceeds. 
If the green fails, its stack, $S_g$, must be empty. Next, the red DFS backtracks from $v$ and tries to reach a vertex, say $w$, 
with $w \neq v$ and $\lv(w) \leq \lv(v)$. If red succeeds, green moves its center of activity to $v$; however, it does not push $v$ onto $S_g$, 
since it has already backtracked from $v$. In addition, it adds $v$ to $G$ and edge $(v_g, v)$ to $T_g$, and DDFS proceeds. 
If red also fails, then its stack also must be empty and $v$ is indeed the required bottleneck. If so, $v$ is added to both
$R$ and $G$ and edges $(v_r, v)$ and $(v_g, v)$ are added to $T_r$ and $T_g$, respectively, and DDFS halts. 
If DDFS does not find a bottleneck, then eventually the two DFSs must find distinct vertices in $l_0$. 

\begin{theorem}
\label{thm.DDFS}
DDFS accomplishes the objectives stated above in the required time.
\end{theorem}

\begin{proof}
The main difference between a usual DFS and the two DFSs implemented in DDFS arises when the two DFSs meet at a vertex, say $v$ at layer $l_j$. 
Observe that once one DFS reaches a vertex at layer $l$, say, at every future point, the center of activity of at least one DFS will
be at level $l$ or lower. Therefore, since both DFSs just moved from higher layers to layer $l_j$, no other vertices at layer $l_j$ or 
lower have yet been explored. Therefore, if $v$ is not a bottleneck, there is an alternative path that reaches layer $l_j$ or lower and which has
not been explored so far. Since at this point both DFSs will consider all ways of finding such an alternative,
they will find one and DDFS will proceed. On the other hand, if $v$ is a bottleneck, there is no such alternative, and after considering all 
possibilities, both stacks will become empty and $v$ will be declared the bottleneck.

If there is no bottleneck in $H$, by the arguments given above, the two DFSs will not get stuck at any vertex. Hence, one of them must reach a 
vertex at layer $l_0$, say $v$. Since $v$ is also not a bottleneck, even if the two DFSs meet at $v$, one of them will find a way of reaching 
another vertex at layer $l_0$.

Clearly, the only edges searched by DDFS are those in $E_b$ in the first case and $E_p$ in the second. Furthermore, each such edge is examined 
at most twice, once in the forward search and once during backtrack. Hence DDFS accomplishes the stated objectives in the required time.
\end{proof}

\section{The algorithm}
\label{sec.alg}

The MV algorithm can be viewed as an intertwining of an alternating BFS (similar to the one used for bipartite graphs) with a number of DDFSs.
The first part of the algorithm for a phase finds the evenlevels and oddlevels of vertices and marks the graph in an appropriate
manner. This part is organized in search levels, the first one being search level 0. Let $j_m = (l_m - 1)/2$, where $l_m$ is the length of 
a minimum length augmenting path in $G$. Then during search level $j_m$, a maximal set of augmenting paths of length $l_m$ is found
and the current phase terminates. If $l_m = \infty$, i.e., there are no augmenting paths, the algorithm will reach a search level at which it 
has found the minlevels and maxlevels of all vertices reachable via alternating paths from the unmatched vertices. At this point it will
halt and output the current matching, which will be maximum.

\subsection{Synchronization of events, and procedures MIN and MAX}
\label{sec.sync}

As stated in the Introduction, a key idea in the MV algorithm is the precise synchronization of events -- we describe this next.
For convenience, we will assume that at the beginning of a phase, all vertices are assigned a temporary minlevel of $\infty$.
During search level $i$ procedure MIN finds all vertices, $v$, having $\mn(v) = i+1$ and assigns these vertices their minlevels.
For each edge MIN scans, it also determines whether it is a prop or a bridge.
After MIN is done, procedure MAX finds all vertices, $v$, having $\t(v) = 2i+1$ and assigns these vertices their maxlevels; their minlevels
are at most $i$ and are already known. See Algorithm \ref{alg} for a summary of the main steps.

Observe that in procedure MIN, if the predicate ``$\mn(v) \geq i+1$'' is true, then either the minlevel of $v$ is still $\infty$ or it has
already been found to be $i+1$; even in the  latter case, $u$ is made a predecessor of $v$.


\bigskip

\noindent

\fbox{
\begin{algorithm}{\label{alg} \ \ \ \ \ \ \ \ \ \ \ \  At search level $i$:}

\bigskip

\step
\label{step1}
{\bf MIN:}  \\

{\bf For} each level $i$ vertex, $u$, search along appropriate parity edges incident at $u$. \\

\begin{description}
\item

{\bf For} each such edge $(u, v)$, if $(u, v)$ has not been scanned before {\bf then}

\begin{description}
\item
{\bf If} $\mn(v) \geq i+1$ {\bf then} \\

\begin{description}
\item
$\mn(v) \la i+1$

\item
Insert $u$ in the list of predecessors of $v$.

\item
Declare edge $(u, v)$ a prop.

\end{description}

\item
{\bf Else} declare $(u, v)$ a bridge.

\end{description}

\item
{\bf End For}

\end{description}

{\bf End For}

\bigskip

\step
{\bf MAX:}  \\

Find the support of each bridge of tenacity $2i+1$ using DDFS. \\
{\bf For} each vertex $v$ in the support: \\  


\begin{description}
\item
$\mx(v) \la 2i+1 - \mn(v)$

\end{description}

{\bf End For}
\medskip

\end{algorithm}
}

\bigskip


In each search level, procedure MIN executes one step of alternating BFS as follows. If $i$ is even (odd), it searches from all vertices, 
$u$, having an evenlevel (oddlevel) of $i$ along incident unmatched (matched) edges, say $(u, v)$. If edge $(u, v)$ has not been scanned before, 
MIN will determine if it is a prop or a bridge. If $v$ already been assigned a minlevel of at most $i$, then $(u, v)$ is a bridge.
Otherwise, $v$ is assigned a minlevel of $i + 1$, $u$ is declared a predecessor of $v$ and edge $(u, v)$ is 
declared a prop. Note that if $i$ is odd, $v$ will have only one predecessor -- its matched neighbor. But if $i$ is even, $v$ may have one or more
predecessors. 

The algorithm also constructs for each odd number $t$, the list of bridges of tenacity $t$, $B(t)$. For each edge that is declared a bridge, 
as soon as the appropriate levels of its two endpoints are known, as stated in Definition \ref{ref.tenacity}, its tenacity is ascertained and 
it is inserted into the appropriate list. Lemma \ref{lem.know} below proves that by the end of execution of procedure MIN at search level 
$i$, the algorithm would have determined the tenacity of all bridges of tenacity $2i + 1$.

At this point, procedure MAX uses DDFS to find the support of each bridge in the set $B(2i+1)$. This yields all vertices of tenacity $2i+1$ and 
their maxlevel can be calculated as indicated in Algorithm \ref{alg}. Knowing maxlevels of these vertices helps the algorithm determine the
tenacity of bridges incident at them.

\begin{figure}[ht]
\begin{minipage}[b]{0.5\linewidth}
\centering
\includegraphics[width=\textwidth]{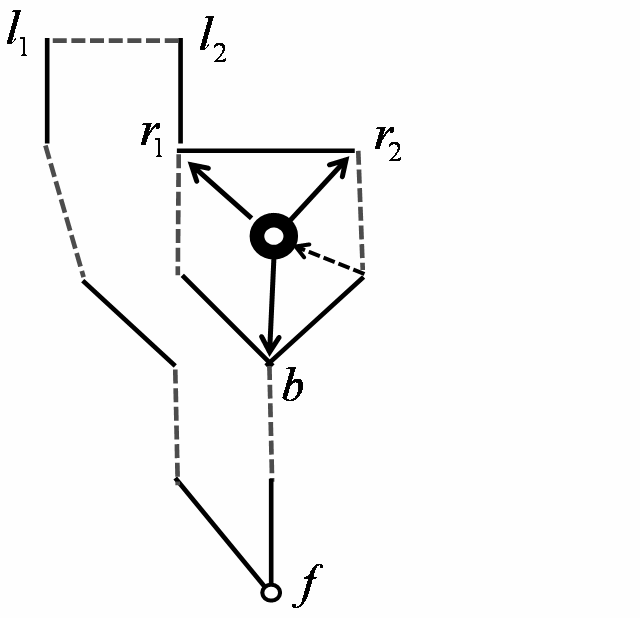}
\caption{A new petal-node is created after DDFS on bridge $(r_1, r_2)$.}
\label{fig.DDFSG}
\end{minipage}
\hspace{0.5cm}
\begin{minipage}[b]{0.5\linewidth}
\centering
\includegraphics[width=\textwidth]{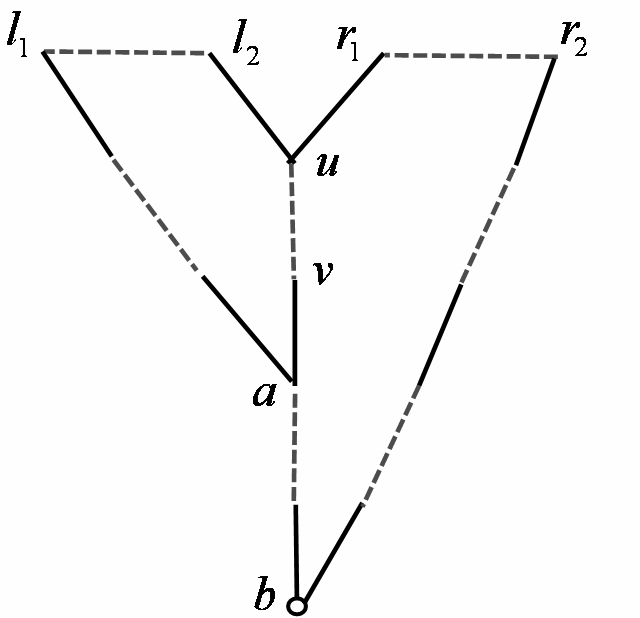}
\caption{Vertices $u$ and $v$ are in the support of both bridges of tenacity of 11.}
\label{fig.disjoint}
\end{minipage}
\end{figure}

\subsection{Running DDFS on $G$}
\label{sec.DDFSG}

We now describe how DDFS is executed on the given graph, $G$; as stated in Section \ref{sec.DDFS}, the graph $H$ is a considerably 
simplified version of $G$. For graph $H$, we specified its vertices, including a layer for each vertex, and its edges. We will specify the
vertices and edges gradually below; however, at the outset we state that the layer of any vertex will be its minlevel.

In the graph of Figure \ref{fig.DDFSG}, MAX will perform DDFS on bridge $(r_1, r_2)$, of tenacity 9, during search level 4, by
starting two DFSs at vertices $r_1$ and $r_2$, respectively. We first state a preliminary rule for determining the edges of the
corresponding graph $H$ -- the preliminary rule will suffice for this first DDFS. If the center of activity of a DFS is at $u$ then it 
must search along all edges $(u, v)$ where $v$ is a predecessor of $u$.

Clearly, this DDFS will terminate with the bottleneck $b$. It will visit the four vertices which constitute the support of bridge 
$(r_1, r_2)$; observe that $b$ is not in the support of bridge $(r_1, r_2)$. These four vertices form the blossom $\CB_{b, 9}$. 
However, in general, DDFS may not find an entire blossom but only a part of it. The set of vertices it does find will be called
a {\em petal}; a blossom, in general, is a union of petals. Thus the four vertices of tenacity 9 are said to belong to the new petal.

At this point, the algorithm creates a new node, called a {\em petal-node}; this has the shape of a donut in 
Figure \ref{fig.DDFSG}. The four vertices of the new petal point to the petal-node; to avoid cluttering Figure \ref{fig.DDFSG}, only one
vertex is pointing to the petal-node. The bottleneck, $b$, is called the {\em bud} of the petal. 
The new petal-node points to the two endpoints of its bridge, $r_1$ and $r_2$, and to its bud, $b$; the reason for the former will be
clarified in Section \ref{sec.finding} and for the latter below.

\begin{figure}[ht]
\begin{minipage}[b]{0.5\linewidth}
\centering
\includegraphics[width=\textwidth]{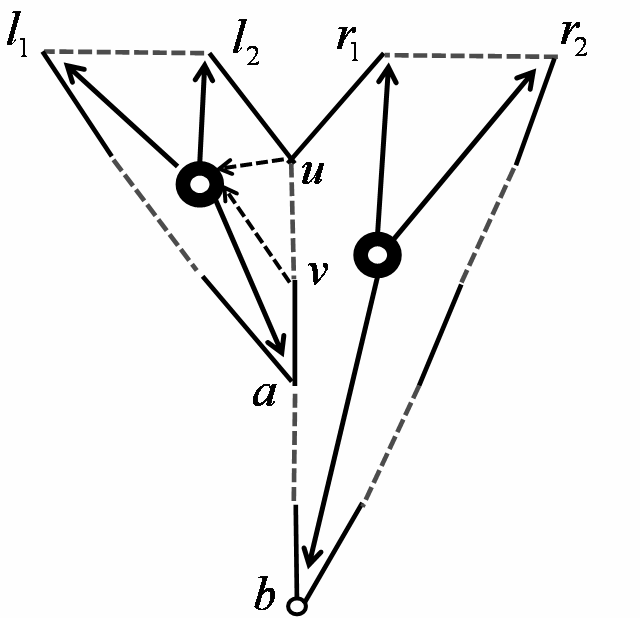}
\caption{DDFS is performed on the left bridge first, then right.}
\label{fig.left}
\end{minipage}
\hspace{0.5cm}
\begin{minipage}[b]{0.5\linewidth}
\centering
\includegraphics[width=\textwidth]{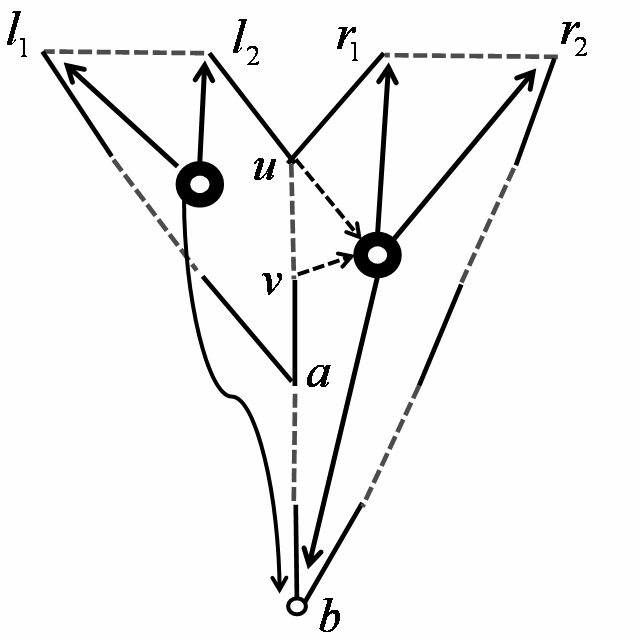}
\caption{DDFS is performed on the right bridge first, then left.}
\label{fig.right}
\end{minipage}
\end{figure}

\definition{(The bud of a vertex)}
If vertex $v$ is in a petal and the bud of this petal is $b$ then will say that the {\em bud of} $v$ is $b$, written as
$\bd(v) = b$; $\bd(v)$ is undefined if $v$ is not in a petal. Next we define the function $\bds(v)$. If $v$ is not in any petal, 
$\bds(v) = v$ else $\bds(v) = \bds(\bd(v))$.

Observe that the bud of a petal is always an outer vertex. Unlike $\base(v)$ which was defined graph-theoretically and independent of the run 
of the algorithm, $\bd(v)$ depends on the manner in which the algorithm breaks ties; for examples, see below.  
Additionally, $\bds(v)$ not only depends on the particular run of the algorithm, it keeps changing as the algorithm proceeds. At any point
in the algorithm, its latest value is used.

In Figure \ref{fig.DDFSG}, a second DDFS is performed on bridge $(l_1, l_2)$, of tenacity 11, during search level 5. To describe this DDFS,
we need to state the complete rule for defining edges of the corresponding graph $H$: if the center of activity of a DFS is at $u$ and it 
searches along edge $(u, v)$, where $v$ is a predecessor of $u$, then DFS must move the center of activity to $\bds(v)$. Thus, when the 
DFS which starts at $l_2$ searches along edge $(l_2, r_1)$, it moves the center of activity to $b$. It does so by following the pointer 
from $r_1$ to its petal-node and from the petal-node to the bud of this petal-node. In the process, the center of activity has jumped down 
more than one layer. The edges of $H$ were allowed to jump down an arbitrary number of layers in order to model this. 

This DDFS will end with bottleneck $f$. The new petal is precisely the support of bridge $(l_1, l_2)$ and consists of the eight
vertices of tenacity 11 in Figure \ref{fig.DDFSG}, which includes $b$. Once again, a new petal-node is created and these eight
vertices point to it. In addition, the petal-node points to $l_1, l_2$ and to $f$.
Observe that the blossom $\CB_{f, 11}$ consists of these eight vertices and the four vertices of blossom $\CB_{b, 9}$.

Next consider the graph of Figure \ref{fig.disjoint} which has two bridges of tenacity 11, $(l_1, l_2)$ and $(r_1, r_2)$.
Observe that vertices $u$ and $v$ are in the support of both these bridges. Hence, the support of bridges need not be disjoint.
Observe also that the base of these vertices is not $a$ but $b$; note that $\t(a) = 11$. Moreover, there is only one blossom in this 
graph, i.e., $\CB_{b, 11}$.

\begin{figure}[ht]
\begin{minipage}[b]{0.5\linewidth}
\centering
\includegraphics[width=\textwidth]{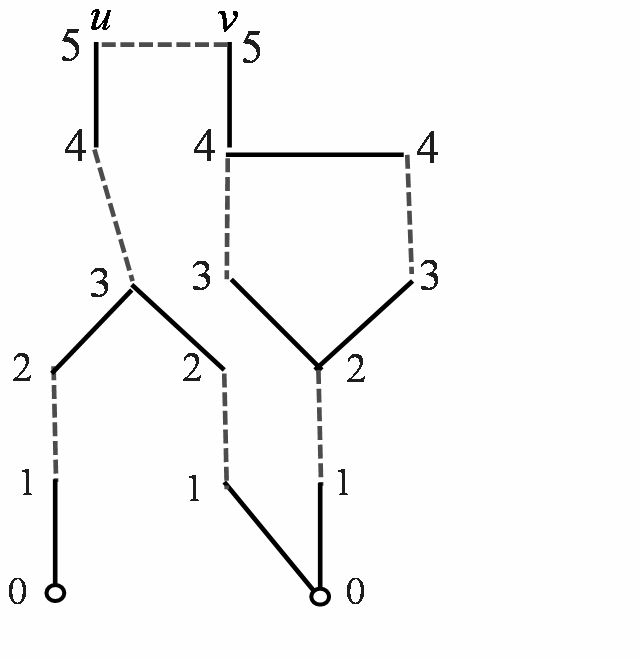}
\caption{DDFS on the bridge of tenacity 11 ends with the two unmatched vertices.}
\label{fig.2paths}
\end{minipage}
\hspace{0.5cm}
\begin{minipage}[b]{0.5\linewidth}
\centering
\includegraphics[width=\textwidth]{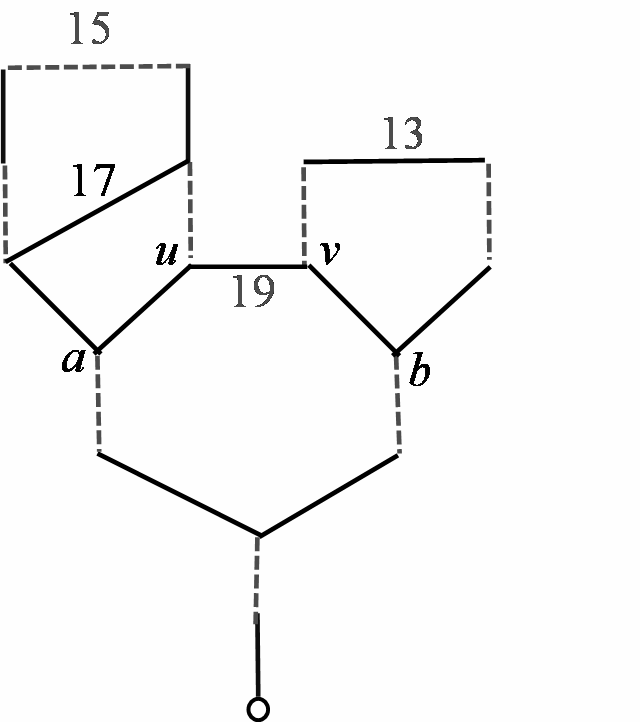}
\caption{DDFS performed on bridge $(u, v)$ starts the two DFSs at $a$ and $b$, respectively.}
\label{fig.ten19}
\end{minipage}
\end{figure}

MAX will perform DDFS on these two bridges in arbitrary order during search level 5. Figure \ref{fig.left} shows the result of performing 
DDFS on $(l_1, l_2)$ before $(r_1, r_2)$. The first DDFS will end with bottleneck $a$. The new petal is precisely the support of 
$(l_1, l_2)$, consisting of six vertices of tenacity 11, including $u$ and $v$. The second DDFS, performed on $(r_1, r_2)$, will end with 
bottleneck $b$ and the new petal is precisely the difference of supports of $(r_1, r_2)$ and $(l_1, l_2)$, i.e., the remaining eight vertices 
of tenacity 11, including $a$. 
During this DDFS, when the DFS starting at $r_1$ searches along edge $(r_1, u)$, it realizes that $u$ is already in a petal and it jumps
to $a$, the bud of this petal. This ensures that a vertex is included in at most one petal. In general, at the end of DDFS on
bridge $(u, v)$, the new petal will be the support of $(u, v)$ minus the supports of all bridges processed thus far in this search level. 

Figure \ref{fig.right} shows the result of performing DDFS on $(r_1, r_2)$ before $(l_1, l_2)$. The first petal is the support of $(r_1, r_2)$,
i.e., 10 vertices of tenacity 11, including $a$, $u$ and $v$. The second petal is the difference of supports of $(l_1, l_2)$ and $(r_1, r_2)$,
i.e., 4 vertices of tenacity 11. In both cases, the union of the petals found is the blossom  $\CB_{b, 11}$.

Figure \ref{fig.2paths} illustrates the second way in which DDFS may end, i.e., instead of a bottleneck, it finds two unmatched vertices;
this happens when DDFS is performed on bridge $(u, v)$ of tenacity 11.

We need to point out one final rule: if DDFS is performed on bridge $(u, v)$, the centers of activity of the two DFSs must start at
$\bds(u)$ and $\bds(v)$. This rule was vacuous so far, but will be applicable while processing bridge $(u, v)$ in Figure \ref{fig.ten19}.
The tenacity of this bridge is 19 and it will be processed by MAX in search level 9. At that point in the algorithm, the bridges
of tenacity 15 and 13 would already be processed and $u$ and $v$ will already be in petals, with $\bd(u) = \bds(u) = a$ and 
$\bd(v) = \bds(v) = b$. Hence the centers of activity of the two DFSs will start at $a$ and $b$, respectively. 
Similarly, in Figure \ref{fig.early2}, when DDFS is performed on bridge $(u, v)$ of tenacity 15, $\bds(u) = u$ and
$\bds(v)$ will be the unmatched vertex.

All bridges considered so far had non-empty supports; however, this will not be the case in a typical graph. As an example,
consider the edge of tenacity 17 in Figure \ref{fig.ten19}. Since it does not assign minlevels to either of its endpoints, it is not a prop
and is therefore a bridge. Clearly the support of this bridge is $\emptyset$. DDFS will discover this right away since the $\bds$ of both of 
its endpoints is $a$. Clearly, DDFS needs to be run on all bridges, since that is the only way of determining whether the support of a given bridge
is empty. 

At the end of search level $i = (t-1)/2$, i.e., once MAX is done processing all bridges of tenacity $t$, all blossoms of tenacity $t$ can be identified
as follows. Let $\t(v) = t$ and let $\bds(v) = b$. The proof of this lemma is straightforward and is omitted.

\begin{lemma}
\label{lem.all}
$\base(v) = b$ and the set $S_{b, t}$ defined in 
Definition \ref{def.blossom}, for blossom $\bt$ is precisely $\{ u ~|~ \t(u) = t \ \mbox{and} \ \bds(u) = b \}$.
\end{lemma}

Observe that if $\bds(v)$ is computed at the end of search level $j > i$, then it may not be $b$ anymore. However, it will be an
iterated base of $v$.

\begin{figure}[h]
\begin{center}
\includegraphics[scale = 0.4]{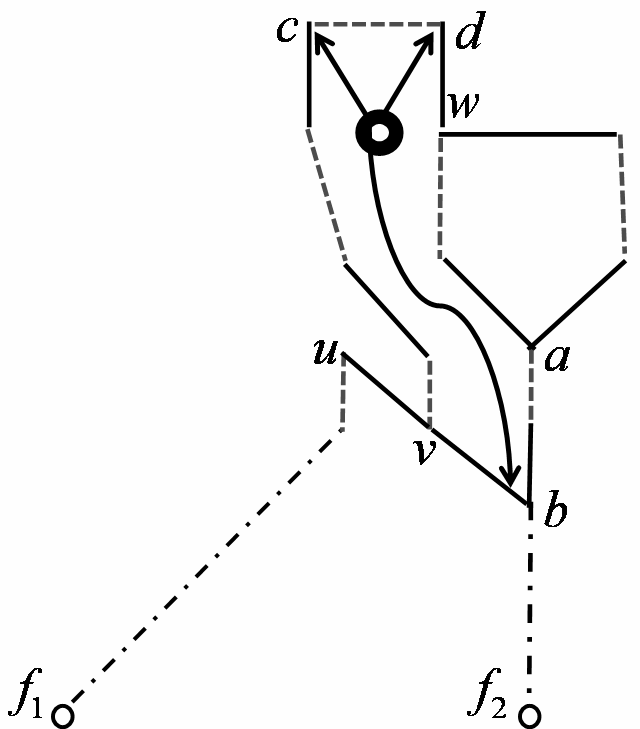}
\caption{Constructing a minimum length augmenting path between unmatched vertices $f_1$ and $f_2$.}
\label{fig.findpath}
\end{center}
\end{figure}

\subsection{Finding the augmenting paths}
\label{sec.finding}

As stated at the beginning of this Section, during search level $j_m$, where $l_m = 2 j_m + 1$ is the length of a minimum length 
augmenting path in $G$, a maximal set of such paths is found; observe that  $l_m$ is also the minimum oddlevel of an unmatched vertex. 
In contrast to previous search levels, in which all DDFSs will end with a bottleneck, in search level $j_m$ DDFS performed on certain 
bridges will end in two unmatched vertices. In this section, we give the procedure that is followed each time such a bridge is encountered; 
in effect, the operations of MAX have to be enhanced during search level $j_m$.

Besides the above-stated event, i.e., at a certain search level, DDFS ends in two unmatched vertices, another possibility is that
$G$ contains two or more unmatched vertices but it has no augmenting paths, i.e., the current matching is maximum (if so, no 
unmatched vertex will have a finite oddlevel). The algorithm will recognize this when it is done assigning minlevels and maxlevels to 
as many vertices as it could, i.e., it has explored the unmatched edges incident at all vertices having a finite evenlevel,
the matched edge incident at all vertices having a finite oddlevel, and has performed DDFS on all bridges of finite tenacity that it has
identified.

\subsubsection{Finding one augmenting path}
\label{sec.one}

In Figure \ref{fig.2paths}, when DDFS performed on bridge $(u, v)$ of tenacity 11 at search level 5, it ends with the two unmatched vertices. 
The next task is to find an augmenting path of length 11 between the two unmatched vertices and containing edge $(u, v)$. 
We will describe this via the graph in Figure \ref{fig.findpath}. In this figure, $\mn(u) > \mn(v)$ and hence edge $(u, v)$ is
a bridge. DDFS on this bridge terminates with the two unmatched vertices, $f_1$ and $f_2$. 
At this point, the stack of the DFS starting at $u$ ($v$) contains all the vertices of tenacity $l_m$ that 
lie on the part of the augmenting path from $u$ to $f_1$ ($v$ to $f_2$); observe that the bottom of the latter stack will contain $b = \bd(v)$. 
All vertices of tenacity less than $l_m$ that constitute such an augmenting path are missing; they lie in petals whose $\bds$s, which are
of tenacity $l_m$, sit on the two stacks. The procedure given below will find one such complete augmenting path by recursively ``opening'' 
the nested petals.

Let us show how to construct the path from $v$ to $f_2$. Since $\bd(v) = b$, we first need to find an $\e(b; v)$ path from $v$ to $b$.
The actions are different depending on whether $v$ is outer or inner; in this case $v$ is inner. Therefore, 
$\e(b; v) = \mx(b; v)$, i.e., the path must use the bridge of this petal, which is $(c, d)$. By jumping from $v$ to its petal-node,
the algorithm can get to the endpoints of this bridge. The ``red'' and ``green'' colors on the vertices of this petal, as assigned
by DDFS (see Section \ref{sec.DDFS}), indicate that $v$ was found via the DFS starting at $c$, say this is the red tree. The 
algorithm does a DFS on the red edges of this petal, starting from $c$, and finds a red path to $v$. Also, it does a DFS from $d$ on the green
edges to find a green path to $b$. 

By the rules set above, the DFS that started at $d$ must have skipped to $a = \bd(w)$ while searching along edge $(d, w)$.
Therefore, the green tree yields the ``path'' $d, a, b$. At this point, we need to recursively ``open'' the petal whose 
bud is $a$ and find an $\e(a; w)$ path in it. Again, we ask whether $w$ is outer or inner. This time the answer is ``outer'' and so
we simply keep picking predecessors of vertices until we get from $w$ to $a$. This path is inserted in the right place in the 
``path'' from $d$ to $b$. The path from $b$ to $f_2$ is obtained via the same process: following down predecessors and recursively finding 
paths through any petals that are encountered on the way.

\subsubsection{Finding a maximal set of paths}
\label{sec.maximal}

Section \ref{sec.one} gave the procedure for finding one minimum length augmenting path, say $p$. We now build on it to
find a maximal set of disjoint minimum length augmenting paths; recall that this was the objective of a phase. 
We first show how to identify the set of vertices that cannot be part of an augmenting path that is disjoint from $p$. 
These will be removed from the graph, and MAX will proceed until it encounters another bridge which makes DDFS end with two unmatched vertices.
The processes are repeated until all bridges of tenacity $l_m$ are processed. Since we do not remove any useful vertices, maximality is guaranteed.

First, all vertices of $p$ are removed; together with a vertex, all its incident edges are also removed. As a 
result, we may create vertices that have no more predecessors. Each such vertex $v$ is also removed since the current graph does not
contain a $\mn(v)$ path anymore. This process is continued until every remaining vertex, other than unmatched ones, has a predecessor.

Let us argue that the next time DDFS ends with two unmatched vertices, we are guaranteed to find an augmenting path between them.
The main question is, ``How are we sure that the procedure of Section \ref{sec.one} will be able to find appropriate paths through
blossoms that have lost some of their vertices?''

The structural properties established in Section \ref{sec.blossom} render the answer to this question surprisingly simple. Assume that vertex 
$v \in p$ is in blossom $\CB$ and that vertex $u$ is in the left-over part of this blossom. By Corollary \ref{cor.bases}, $p$ must contain the 
iterated bases of $v$ and the next minimum length augmenting path, using $u$, must also contain the iterated bases of $u$. However, the base, 
$b$, of blossom $\CB$ is an iterated base of both the vertices! Therefore, $b$ will be removed when $p$ is removed and $u$ cannot be on any 
minimum length augmenting path that is disjoint from $p$. Indeed, it is easy to see that the process of iteratively removing vertices that have no
predecessors left will end up removing all of $\CB$.

\section{Proof of correctness and running time}
\label{sec.proof}

\begin{lemma}
\label{lem.un-bridge}
Let $(u, v)$ be an unmatched edge that is a bridge with $\t(u, v) \leq l_m$.
Then $\t(u) \leq \t(u, v)$.
\end{lemma}

\begin{proof}
Let $p$ be an $\e(v)$ path, starting at unmatched vertex $f$, say. If $u$ does not occur on $p$ then $p \circ (v, u)$ is an odd alternating path to $u$.
Therefore $\o(u) \leq \e(v) = 1$ and the lemma follows. If $u$ does occurs on $p$ and $u$ is odd w.r.t. $p$ then
$\o(u) < \e(v)$ and again we are done.

Next suppose that $u$ occurs on $p$ and $u$ is even w.r.t. $p$. Since $(u, v)$ is not a prop, $p[f \ \mbox{to} \ u] \circ (u, v)$ is
not an $\o(v)$ path, and an $o(v)$ path must be shorter. Let $q$ be such a path, starting at unmatched vertex $f'$, say.
If $u$ does not lie on $q$ then $q \circ p[v \ \mbox{to} \ u]$ is an odd alternating path to $u$. Furthermore, $p[f \ \mbox{to} \ u]$
is an even alternating path to $u$. Therefore $\t(u) \leq \t(u, v)$. 

Finally consider the case that $u$ lies on $q$. If $u$ is odd w.r.t. $q$ then again $\o(u)$ is small enough and we are done.
If $u$ is even w.r.t. $q$ then $|p[f \ \mbox{to} \ u]| > |q| > |q[f' \ \mbox{to} \ u]|$. Now $q$ must intersect $p(u \ \mbox{to} \ v]$ and the 
first intersection must be at a vertex that is even w.r.t. $p$, otherwise there is an even alternating path to $v$ that is shorter than $\e(v)$.
Appropriate parts of $q$ and $p$ now yield an odd alternating path to $u$ of length less than $\e(v)$ and the lemma follows. 
\end{proof}

\begin{lemma}
\label{lem.uses}
Let $(u, v)$ be a bridge such that $\t(v) = \t(u, v) \leq l_m$.
Then the expression for $\t(u, v)$ uses $\mn(v)$.
\end{lemma}

\begin{proof}
First assume that $(u, v)$ is matched. Since it is not a prop, $u$ and $v$ must both be inner vertices with $\o(u) = \o(v) = i = \mn(v)$, and
the expression for $\t(u, v)$ uses $\o(v)$.

Next assume that $(u, v)$ is unmatched.
Therefore the equality $\t(v) = \t(u, v)$ implies
\[ \o(v) + \e(v) = \e(u) + \e(v) + 1 \ \ \ \implies \ \ \ \o(u) = e(v) + 1 .\] 
Now, if $\o(v) = \mn(v)$ then $u$ must be a predecessor of $v$, contradicting the hypothesis that $(u, v)$ is a bridge.
Therefore, $\o(v) = \mx(v)$. On the other hand, the expression for $\t(u, v)$ uses $\e(v)$ and hence it uses $\mn(v)$.
\end{proof}

\begin{lemma}
\label{lem.know}
For each bridge $(u, v)$ having tenacity $2i+1 \leq l_m$, the algorithm will determine that $(u, v)$ is a bridge and that
its tenacity is $2i+1$ by the end of execution of procedure MIN at search level $i$.
\end{lemma}

\begin{proof}
First assume that $(u, v)$ is matched. As argued in Lemma \ref{lem.uses}, $\o(u) = \o(v) = i$.
Therefore, during search level $i$, MIN will determine that $(u, v)$ is a bridge and that its tenacity is $2i+1$.

Next assume that $(u, v)$ is unmatched and that $\e(u) \leq \e(v)$ and that $(u, v)$ is first scanned from $u$. 
Clearly $\e(u) \leq i$. Since $(u, v)$ is not a prop, $v$ must already have a minlevel and $\mn(v) \leq i$. 
By Lemma \ref{lem.un-bridge}, $\t(v) \leq \t(u, v)$. Now if $\t(v) < \t(u, v)$, both levels of $v$ are already known 
and hence $\t(u, v)$ can be determined. If $\t(v) = \t(u, v)$ then by Lemma \ref{lem.uses} the expression for $\t(u, v)$ uses $\mn(v)$
and hence $\t(u, v)$ can be determined. Further notice that since $\t(u, v) = 2i+1$, $\mn(v)$ must be $i$.
\end{proof}

\begin{theorem}
\label{thm.levels}
For each vertex $v$ such that $\t(v) \leq l_m$, Algorithm \ref{alg} assigns $\mn(v)$ and $\mx(v)$ correctly. 
\end{theorem}

\begin{proof}
We will show by strong induction on $i$ that at search level $i$, Algorithm \ref{alg} assigns a minlevel of $i+1$ to exactly the set of 
vertices having this minlevel and it assigns appropriate maxlevels to exactly the set of vertices having tenacity $2i+1$.
The basis, i.e, $i = 0$, is obvious.

To prove the induction step, assume that the hypothesis is true for all search levels less than $i$.
Let $\mn(v) = i+1$, let $p$ be a $\mn(v)$ path, starting at unmatched vertex $f$, and let $(u, v)$ be the last edge on $p$. It is easy to 
see that $u$ must be BFS-honest w.r.t. $p$; if not, then $v$ must occur on the shorter path to $u$, which contradicts the assumption 
that $\mn(v) = i+1$. Therefore, if $|p[f \ \mbox{to} \ u]|$ is odd (even), $\o(u) = i$ ($\e(u) = i$). By the induction hypothesis,
$u$ must be assigned this level, regardless of whether it is the minlevel or maxlevel of $u$. Therefore, on searching along an appropriate 
parity edge incident at $u$, MIN will find $v$. By the induction hypothesis, at the start of search level $i$, the minlevel of $v$ is not set.
Therefore, when edge $(u, v)$ is examined, it is either still not set or it is set to $i+1$ 
through some other edge scanned in the current search level (the latter case happens only if $\o(v) = i+1$). 
Hence, while searching along edge $(u, v)$, $v$ will be assigned its correct minlevel, $u$ will be declared a predecessor of $v$ and 
$(u, v)$ will be declared a prop.

Next, assume $t(v) = 2i+1$ and let $p$ be a $\mx(v)$ path. By Theorem \ref{thm.bridge} there exists a unique bridge of tenacity $t$
on $p$, say it is $(a, b)$. By Lemma \ref{lem.know}, by the end of MIN in search level $i$, $(a, b)$ must be inserted in the list $B(2i+1)$.
Therefore, MAX will call DDFS with this bridge to find its support. By the induction hypothesis, $\mx(v)$ was not set in any of the 
previous search levels\footnote{The importance of this subtle point, which is related to the idea of ``precise synchronization of
events'' is explained below with the help of Figures \ref{fig.early1} and \ref{fig.early2}.
}.
However, it may be set in the current search level; if so, $v$ already belongs to a petal. If it is not set, DDFS will 
reach $v$ and assign it its maxlevel.
\end{proof}

\begin{figure}[ht]
\begin{minipage}[b]{0.5\linewidth}
\centering
\includegraphics[width=\textwidth]{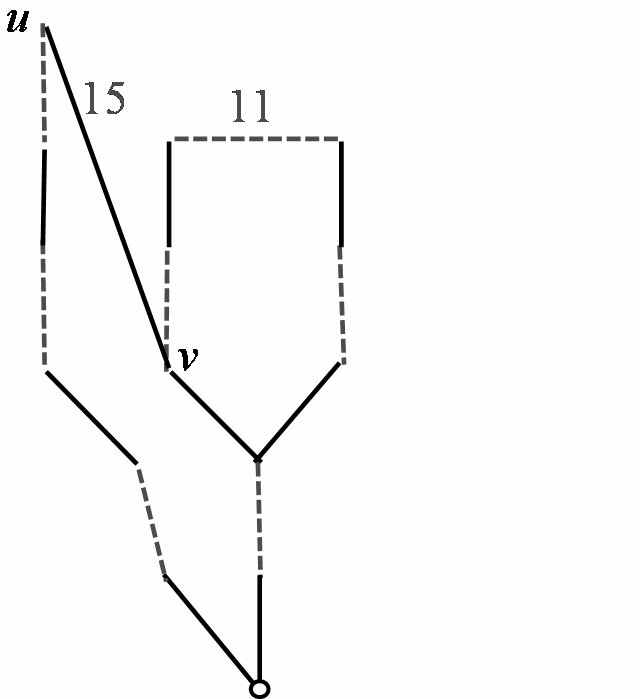}
\caption{Can DDFS be performed on bridge $(u, v)$ at search level 6?}
\label{fig.early1}
\end{minipage}
\hspace{0.5cm}
\begin{minipage}[b]{0.5\linewidth}
\centering
\includegraphics[width=\textwidth]{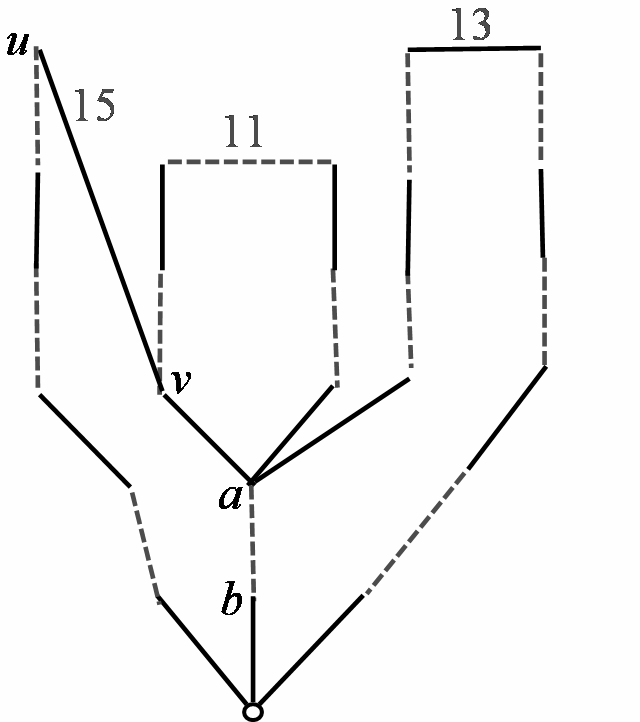}
\caption{If so, vertices $a$ and $b$ will get wrong tenacities.}
\label{fig.early2}
\end{minipage}
\end{figure}

In Figure \ref{fig.early1}, the algorithm determines that $(u, v)$ is a bridge of tenacity 15 at search level 6.
However, according to Algorithm \ref{alg}, DDFS has to be performed on $(u, v)$ at search level 7. The question arises,
``Why wait till search level 7; why not perform DDFS on $(u, v)$ when procedure MAX is run at search level 6?'' To clarify this, let us change 
the algorithm so it runs DDFS on an edge as soon as its tenacity and its status as a bridge have been determined. Assume further
that among the various bridges ready for processing, ties are broken arbitrarily\footnote{By making the example given in Figure 
\ref{fig.early2} slightly bigger, one can easily ensure that there are no such ties.}.

Now consider the enhanced graph of Figure \ref{fig.early2}, in which vertices $a$ and $b$ are clearly in the support of the bridge of 
tenacity 13 and hence have tenacity 13. Assume that when MAX is run at search level 6 the bridge of tenacity 15 is processed first. 
Since the tenacities of vertices $a$ and $b$ are not set yet, DDFS will visit them and assign them a of tenacity 15, which would be incorrect. 
Observe that the correctness of MAX crucially depends on assigning tenacities to each edge of tenacity less than $2i+1$ before
processing bridges of tenacity $2i+1$, i.e., the precise manner in which events are synchronized in Algorithm \ref{alg}.

\begin{theorem}
\label{thm.time}
The MV algorithm finds a maximum matching in general graphs in time $O(m \sqrt{n})$ on the RAM model and 
$O(m \sqrt{n} \cdot \alpha(m, n))$ on the pointer model, where $\alpha$ is the inverse Ackerman function.
\end{theorem}

\begin{proof}
Through arguments made so far, it should be clear that each of the procedures of MIN, MAX, finding augmenting paths, and removing vertices
after each augmentation will examine each edge a constant number of times in each phase. The only operation that remains is
that of computing $\bds$ during DDFS. This can either be implemented on the pointer model by using Tarjan's set union algorithm \cite{Tarjan}, 
which will take $O(m \cdot \alpha(m, n))$ time per phase, or on the RAM model by using Gabow and Tarjan's linear 
time algorithm for a special case of set union \cite{GTarjan}, which will take $O(m)$ time per phase. Since $O(\sqrt{n})$ phases
suffice for finding a maximum matching \cite{Karp,Karzanov}, the theorem follows.
\end{proof}

A question arising from Theorem \ref{thm.time} is whether there is a linear time implementation of $\bds$ on the pointer model.
\cite{MV} had claimed, without proof, that because of the special structure of blossoms, path compression itself suffices, together with a 
charging argument that assigns a constant cost to each edge. This claim could not be verified at the time of writing \cite{va.matching},
so it was left as an open problem in that paper. This problem has recently been resolved in the negative. \cite{PV} give
an infinite family of graphs on which path compression in a phase takes time $\Omega(m \alpha(m, n))$.

\section{Equivalence of definitions}
\label{sec.equivalence}

Below we establish equivalence of the two definitions of blossoms. 
Let us start by providing the definition of blossom as given in \cite{va.matching}; we will denote such a blossom of tenacity $t < l_m$ and 
base $b$ by $\bt^o$. Let $v$ be a vertex with $\t(v) \leq t$. 
We will say that an outer vertex $b$ is $\base_{> t} (v)$ if for some positive $k$, $\base^k (v) = b$, $\t(b) > t$,
and $\t(\base^{k-1} (v)) \leq t$. Then
\[ \bt^o = \{ v ~|~ \t(v) \leq t \ \mbox{and} \ \base_{>t} (v) = b \} .\]

\begin{proposition}
\label{prop}
The two definitions of blossom are equivalent, i.e., $\bt = \bt^o$.
\end{proposition}

\begin{proof}
Let $v \in \bt$. We will show that $v \in \bt^o$ by considering the following three cases. The set $S_{b, t}$ is defined in Definition \ref{def.blossom}.
\begin{enumerate}
\item
$v \in S_{b, t}$. In this case, $\t(v) = t$ and $\base(v) = b$, and therefore $\base_{>t} (v)) = b$. Hence $v \in \bt^o$.
\item
$v \in \CB_{b, t-2}$. In this case, $\t(v) < t$ and for some $k \geq 1$, $\base^k(v) = b$. Clearly, $\base^{k-1} (v) \in \CB_{b, t-2}$
and therefore $\base_{>t} (v) = b$. Hence $v \in \bt^o$.
\item
$v \in \CB_{u, t-2}$ and $u \in S_{b, t}$. In this case, $\t(v) < t$, $\t(u) = t$, $\base(u) = b$, and for some $k \geq 1$, $\base^k(v) = u$. Therefore,
$\base^{k+1} (v) = b$ and $\base_{>t} (v)) = b$. Hence $v \in \bt^o$.
\end{enumerate}

Next, let $v \in \bt^o$. Once again we will consider three cases to show that $v \in \bt$.
\begin{enumerate}
\item
$\t(v) = t$. In this case, $\base(v) = b$ and therefore $v \in S_{b, t}$. Hence $v \in \bt$.
\item
$\t(v) < t$, for some $k \geq 1$, $\base^k(v) = b$ and $\t(\base^{k-1} (v)) < t$. In this case, $v \in \CB_{b, t-2}$. Hence $v \in \bt$.
\item
$\t(v) < t$, for some $k \geq 1$, $\base^k(v) = b$ and $\t(\base^{k-1} (v)) = t$. Let $\base^{k-1} (v) = u$. Then,
$\t(u) = t$ and $\base(u) = b$. Therefore, $u \in S_{b, t}$ and $v \in \CB_{u, t-2}$. Hence $v \in \bt$.
\end{enumerate}
\end{proof}

\section{Epilogue}
\label{sec.discussion}

In summary, the main new task to be accomplished in non-bipartite graphs, beyond bipartite graphs, is to find maxlevels of vertices.
The process of finding minlevels of vertices is very much the same as the process of finding levels
of vertices from all unmatched vertices in one of the bipartitions in a bipartitie graph, i.e., an alternating BFS, as carried out by the
procedure MIN. And its proof of correctness is straightforward  -- the ``agent'' that is responsible for assigning vertex $v$ its minlevel is easily 
seen to be one of the neighbors of $v$ (see the proof of Theorem \ref{thm.levels}).

What is the ``agent'' that is responsible for assigning a vertex its maxlevel? The answer is far from straightforward and 
is established in Theorem \ref{thm.bridge}. In a sense, our motivation for finding all the structural facts given before this theorem was precisely to 
prove this theorem. Once these structural facts were found, it became clear that the MV algorithm was ``walking'' on precisely this structure.
And with this came the realization that this structure was also the key to a conceptual description of the algorithm.
We hope this viewpoint will help with a better understanding of the paper.

\section{Acknowledgments}
Bob Tarjan initiated the idea of giving talks on the MV algorithm (which unfortunately had remained a ``black box'' result for 33 years) and also
encouraged me to write a fresh manuscript. I had valuable discussions with Bob Tarjan and Laci Lovasz, and I wish to thank them both.
I also wish to thank Jugal Garg and Rakshit Trivedi, with special thanks to Ruta Mehta, for allowing me to sound out my ideas as they were evolving,
and to Thomas Hansen, Alexey Pechorin and Uri Zwick for carefully reading the paper and providing valuable comments.
\bibliography{kelly} 
\bibliographystyle{alpha}

\end{document}